\documentclass[draftcls,onecolumn]{IEEEtran}

\usepackage{cite}
\usepackage{pbox}
\usepackage{color}
\usepackage{pstricks}
\usepackage{graphicx}
\usepackage{amsmath}
\usepackage{amssymb}
\usepackage{amsthm}
\usepackage{mathrsfs}
\usepackage{mathtools}
\usepackage{algorithm,algpseudocode}
\usepackage{psfrag}
\usepackage[caption=false,font=footnotesize]{subfig}
\usepackage[margin=1in,dvips]{geometry}
\usepackage[ulem=normalem]{changes}
\usepackage{tikz,pgfplots}
\usepackage{url}
\pgfplotsset{compat=newest}
\usetikzlibrary{matrix,positioning,calc,fit,backgrounds}

\newrgbcolor{ColorLink}{0 .4 .6}
\newrgbcolor{ColorCite}{0 .6 .4}

\usepackage[plainpages=false,pdfpagelabels,colorlinks=true,linkcolor=ColorLink,citecolor=ColorCite]{hyperref}

\let\oldbrace\{
\def\{{\oldbrace\kern0.5pt}

\def\diag{\mathop{\rm diag}\nolimits}%
\def\rank{\mathop{\rm rank}\nolimits}%
\def\nullspace{\mathop{\rm null}\nolimits}%
\def\Span{\mathop{\rm span}\nolimits}%
\newcommand{\nn}{\nonumber}

\newcommand{\MAC}{\mathsf{MAC}}
\newcommand{\LMAC}{\mathsf{LMAC}}
\newcommand{\CF}{\mathsf{CF}}

\newcommand{\q}{\mathsf{q}}
\newcommand{\Fq}{\Field_\q}
\newcommand{\Fqh}{\hat{\Field}_\q}

\newcommand{\Ac}{\mathcal{A}}

\newcommand{\Cc}{\mathcal{C}}

\newcommand{\Ec}{\mathcal{E}}

\newcommand{\Ic}{\mathcal{I}}

\newcommand{\Kc}{\mathcal{K}}
\newcommand{\Lc}{\mathcal{L}}
\newcommand{\Mc}{\mathcal{M}}

\newcommand{\Sc}{\mathcal{S}}
\newcommand{\Tc}{\mathcal{T}}

\newcommand{\Xc}{\mathcal{X}}
\newcommand{\Yc}{\mathcal{Y}}
\newcommand{\Zc}{\mathcal{Z}}

\newcommand{\lv}{\boldsymbol{l}}

\newcommand{\mv}{\boldsymbol{m}}

\newcommand{\pen}{{P_e^{(n)}}}

\newcommand{\aep}{{\mathcal{T}_{\epsilon}^{(n)}}}
\newcommand{\aepvar}{{\mathcal{T}_{\epsilon'}^{(n)}}}

\newcommand{\Rh}{{\hat{R}}}

\newcommand{\Wh}{{\hat{W}}}

\newcommand{\uh}{{\hat{u}}}

\newcommand{\Rt}{{\tilde{R}}}

\def\d{\delta}
\def\e{\epsilon}

\let\P\relax
\DeclareMathOperator\P{\sf P}

\newcommand{\U}{\mathrm{Unif}}

\newcommand{\Integer}{\mathbb{Z}}
\newcommand{\Field}{\mathbb{F}}

\DeclarePairedDelimiter{\ceil}{\lceil}{\rceil}

\newrgbcolor{lineblue}{.2 .6 .9}
\newrgbcolor{linered}{.9 .3 .3}
\newrgbcolor{linegreen}{.3 .8 .3}
\newrgbcolor{linepurple}{.75 .5 .75}

\newrgbcolor{LineGray}{.6 .6 .6}
\newrgbcolor{AxisGray}{.3 .3 .3}
\newrgbcolor{LineBlue}{.2 .6 .9}
\newrgbcolor{LineRed}{.9 .3 .3}
\newrgbcolor{LineGreen}{.3 .8 .3}
\newrgbcolor{LinePurple}{.8 .3 .8}

\newcommand{\tra}{\mathsf{T}}

\newtheorem{theorem}{Theorem}
\newtheorem{lemma}{Lemma}
\newtheorem{corollary}{Corollary}

\theoremstyle{definition}
\newtheorem{example}{Example}
\newtheorem{remark}{Remark}

\newcommand{\Rmax}{R_{\text{max}}}

\newcommand{\Rtmax}{\Rt_{\text{max}}}

\newcommand{\mub}{{\sf m}}
\newcommand{\Gb}{\mathsf{G}}

\newcommand{\joint}{\mathrm{joint}}
\newcommand{\seq}{\mathrm{seq}}

\newcommand{\As}{\mathsf{A}}
\newcommand{\Bs}{\mathsf{B}}
\newcommand{\Cs}{\mathsf{C}}
\newcommand{\Ds}{\mathsf{D}}
\newcommand{\Is}{\mathsf{I}}

\newcommand{\Ib}{\mathsf{I}}

\newcommand{\Qb}{\mathsf{Q}}
\newcommand{\Mub}{\mathsf{M}}
\newcommand{\zerob}{\boldsymbol{0}}
\newcommand{\Ab}{\boldsymbol{a}}
\newcommand{\ab}{\boldsymbol{a}}

\newcommand{\hb}{\boldsymbol{h}}
\newcommand{\eb}{\textbf{\sf e}}
\newcommand{\Tr}{\mathsf{T}}

\newcommand{\Fs}{\mathsf{F}}

\IEEEoverridecommandlockouts

\begin{document}

\title{Towards an Algebraic Network Information Theory: Simultaneous Joint Typicality Decoding}

\author{Sung Hoon Lim, Chen Feng, Adriano Pastore, Bobak Nazer, Michael Gastpar
\thanks{This paper was presented in part at the 2017 IEEE International Symposium on Information Theory.}
\thanks{Sung Hoon Lim is with the Korea Institute of Ocean Science and Technology, Busan, Korea (e-mail: shlim@kiost.ac.kr).}
\thanks{Chen Feng is with the School of Engineering, The University of British Columbia, Kelowna, BC, Canada (e-mail: chen.feng@ubc.ca).}
\thanks{Adriano Pastore is with the Centre Tecnol\`{o}gic de Telecomunicacions de Catalunya (CTTC/CERCA), Avinguda Carl Friedrich Gauss 7, 08860 Castelldefels, Spain (e-mail: adriano.pastore@cttc.cat).}
\thanks{Bobak Nazer is with the Department of Electrical and Computer Engineering, Boston University, Boston, MA (e-mail: bobak@bu.edu).}
\thanks{Michael Gastpar is with the School of Computer and Communication Sciences, Ecole Polytechnique F\'ed\'erale, 1015 Lausanne, Switzerland (e-mail: michael.gastpar@epfl.ch).}
}

\maketitle

\begin{abstract}
Consider a receiver in a multi-user network that wishes to decode several messages. Simultaneous joint typicality decoding is one of the most powerful techniques for determining the fundamental limits at which reliable decoding is possible. This technique has historically been used in conjunction with random i.i.d.~codebooks to establish achievable rate regions for networks. Recently, it has been shown that, in certain scenarios, nested linear codebooks in conjunction with ``single-user'' or sequential decoding can yield better achievable rates. For instance, the compute--forward problem examines the scenario of recovering $L \le K$ linear combinations of transmitted codewords over a $K$-user multiple-access channel (MAC), and it is well established that linear codebooks can yield higher rates. Here, we develop bounds for simultaneous joint typicality decoding used in conjunction with nested linear codebooks, and apply them to obtain a larger achievable region for compute--forward over a $K$-user discrete memoryless MAC. The key technical challenge is that competing codeword tuples that are linearly dependent on the true codeword tuple introduce statistical dependencies, which requires careful partitioning of the associated error events. 
\end{abstract}


\section{Introduction} \label{sec:intro}
For several decades, decode--forward~\cite{Cover--El-Gamal1979}, compress--forward~\cite{Cover--El-Gamal1979}, and amplify--forward~\cite{Schein--Gallager2000} have served as the fundamental building blocks of transmission strategies for relay networks. These three relaying strategies were initially developed on canonical network models such as the relay channel and diamond relay network using random independent and identically distributed (i.i.d.)~codebooks and joint typicality decoding arguments. Subsequently, these strategies were generalized to $N$-user relay networks~\cite{Kramer--Gastpar--Gupta2005, Lim--Kim--El-Gamal--Chung2011, Yassaee--Aref2011, Hou--Kramer2016, El-Gamal--Kim2011, Lim--Kim--Kim2017, Minero--Lim--Kim2015, Maric--Goldsmith--Medard2012} that also relied upon random i.i.d.~codebooks and joint typicality decoding. 

Beginning with the many-help-one source coding work of K\"orner and Marton~\cite{Korner--Marton1979} followed by a series of recent papers~\cite{Wilson--Narayanan--Pfister--Sprintson2010,Nam--Chung--Lee2010,Nazer--Gastpar2011,Niesen--Whiting2012,Song--Devroye2013,Hong--Caire2013,Ren--Goseling--Weber--Gastpar2014,Bresler--Parekh--Tse2010,Motahari--Gharan--Maddah-Ali--Khandani2014,Niesen--Maddah-Ali2013,Ordentlich--Erez--Nazer2014,Shomorony--Avestimehr2014,Padakandla--Sahebi--Pradhan2016,Krithivasan--Pradhan2009,Krithivasan--Pradhan2011,Wagner2011,Lalitha--Prakash--Vinodh--Kumar--Pradhan2013, Yang--Xiong2014,Philosof--Zamir2009,Philosof--Zamir--Erez--Khisti2011,Wang2012,Padakandla--Pradhan2017,Padakandla--Pradhan2018,He--Yener2014,Vatedka--Kashyap--Thangaraj2015,Xie--Ulukus2014}, it has been observed that random i.i.d.~codebooks may not suffice to attain the capacity region of certain networks. Instead, codes with some form of \textit{algebraic structure}, such as nested linear or lattice codes, can sometimes attain larger rate regions. In the context of relaying, this has led to a fourth relaying paradigm known as compute--forward~\cite{Wilson--Narayanan--Pfister--Sprintson2010,Nam--Chung--Lee2010,Nazer--Gastpar2011,Niesen--Whiting2012,Song--Devroye2013,Hong--Caire2013,Ren--Goseling--Weber--Gastpar2014}. The key idea is that, if all users employ the same linear or lattice codebook, then linear combinations of codewords are themselves codewords, and can often be recovered at higher rates as compared to recovering one (or more) codewords. After the relays recover linear combinations, they forward them to the destinations, which then obtain their desired codewords by solving a system of linear equations. This strategy was originally proposed for Gaussian channels with equal rates and power constraints using (random) nested lattice codes combined with ``single-user'' lattice decoding~\cite{Nazer--Gastpar2011}. It was subsequently generalized to include unequal power constraints and rates as well as sequential decoding~\cite{Ordentlich--Erez--Nazer2014,Nazer--Cadambe--Ntranos--Caire2016,Zhu--Gastpar2017}.

Much of the prior work that demonstrates the rate gains of random linear or lattice codes over random i.i.d.~codes has focused on either binary or Gaussian channels. Inspired by these examples, there is now a 
concerted effort to generalize these results into proof techniques with the objective to develop an \textit{algebraic network information theory} based on codes with algebraic structure. (See the textbook of El Gamal and Kim for the state-of-the-art rate regions for random i.i.d.~codes~\cite{El-Gamal--Kim2011}.) As demonstrated by Padakandla and Pradhan~\cite{Padakandla--Sahebi--Pradhan2016,Padakandla--Pradhan2017,Padakandla--Pradhan2018}, random nested linear codes, when combined with joint typicality encoding and decoding, can be used to generalize the aforementioned examples to discrete memoryless networks. The key insight is that, although a straightforward application of a random linear codebook will lead to a uniform input distribution, joint typicality encoding (i.e. multicoding) can be used to shape a random nested linear codebook to induce any input distribution. This phenomenon was independently discovered in the context of sparse linear codes by Miyake~\cite{Miyake2010}.

In this paper, we develop techniques for bounding the error probability for \textit{simultaneous joint typicality decoding} when used in conjunction with nested linear codebooks. The main technical difficulty is that (exponentially many) competing codewords are linearly dependent on the true codewords, and thus create statistical dependencies that are not handled by classical bounding techniques. We partition error events based on a particular rank criterion, which in turn enables us to characterize the rate penalties that stem from these linear dependencies. We apply our bounds towards deriving an achievable rate region for the general compute--forward problem of recovering $L \leq K$ linear combinations over a $K$-user discrete memoryless MAC. In prior work, we derived an achievable region for the special case of $K = 2$ users and, in the process, generalized technical lemmas from network information theory (e.g., packing, covering, Markov) to apply to nested linear codes~\cite{Lim--Feng--Pastore--Nazer--Gastpar2018}. We employ these lemmas as part of our derivations for the $K \geq 2$ setting, and find that our achievable rate region improves upon our previous results for the $K=2$ case. Overall, simultaneous decoding has played an important role in the development of many results in classical network information theory, and the simultaneous decoding bounds developed herein may also prove useful beyond the compute--forward setting.

The rest of the paper is organized as follows. In the next section, we formally give the problem statement. In Section~\ref{sec:mainresults}, we state our main results on the joint decoding rate region for computing multiple linear combinations (Theorem~\ref{thm:LK-joint}). In Section~\ref{sec:proof}, we give the proof of Theorem~\ref{thm:LK-joint} and finally, in Section~\ref{sec:discussions}, we conclude with some discussions. 

We closely follow the notation in~\cite{El-Gamal--Kim2011}. Let $\Xc$ denote a discrete set and $x^n$ a length-$n$ sequence whose elements belong to $\Xc$. 
We use uppercase letters to denote random variables. For instance, $X$ is a random variable that takes values in $\Xc$. We follow standard notation for probability measures. Specifically, we denote the probability of an event $\Ac$ by $\P\{\Ac\}$ and use $p_X(x)$ to denote probability mass functions (pmf).

For a discrete set $\Xc$, the type of $x^n$ is defined to be $\pi(x | x^n) := \big| \{ i: x_i = x \} \big| /n$ for $x \in \Xc$. Let $X$ be a discrete random variable over $\Xc$ with probability mass function $p_X(x)$. For any parameter $\e \in (0,1)$,  we define the set of $\e$-typical $n$-sequences
 $x^n$ (or the typical set in short)~\cite{Orlitsky--Roche2001} as
$\aep(X) = \{ x^n : | \pi(x|x^n) - p_X(x) | \le \e p_X(x)
\text{ for all } x \in \Xc \}$.  We use $\delta(\e) > 0$ to denote a generic function
 of $\e > 0$ that tends to zero as $\e \to 0$. One notable departure is that we define sets of message indices starting at zero rather than one with shorthand $[n] := \{0,\ldots, n-1\}$. We also define $[1:n] = \{1,\ldots,n\}$ and reserve $\Kc=[1:K]$ to denote the full set of users.

We use the notation $\Fq$ to denote a finite field of order $\q$. We denote deterministic row vectors with lowercase, boldface font (e.g., $\ab \in \Fq^K$). Note that row vectors can also be written as a sequence (e.g., $u^n \in \Fq^n$). We will denote random sequences using uppercase font (e.g., $U^n \in \Fq^n$). Random matrices will be denoted with uppercase, boldface font (e.g., $\mathbf{G} \in \Fq^{n \times \kappa}$) and we will use uppercase, sans-serif font to denote realizations of random matrices (e.g., $\mathsf{G} \in \Fq^{n \times \kappa}$) or deterministic matrices. We denote by $\eb_k\in\Fq^K$ the standard basis (row) vector where the $k$-th element is $1$ and the rest of the elements are all zero.

Define the matrix $\Ib(\Sc)\in\Fq^{|\Sc|\times K}$ as a subset of the identity matrix $\Ib \in \Fq^{K \times K}$ composed of the standard basis vectors $\eb_k$, $k\in\Sc$, i.e., the rows of $\Is(\Sc)$ are $\eb_k$, $k\in\Sc$. Likewise, for any matrix $\As$, we define $\As(\Sc)$ as the submatrix containing only those rows of $\As$ whose index is in $\Sc$, i.e., $\As(\Sc)=\Is(\Sc)\As$. Specifically for vectors, we will frequently use the shorthand $\As_k$ for $\As(\{k\})$.
We denote the \emph{row} span of $\As$ by $\Span(\As)$ as well as its nullspace by $\nullspace(\As)$. 
Throughout the paper, we assume that all rates $R_k$, $k\in\Kc$ are non-negative and are subject to constraints $R_k\ge 0$.

We define an \emph{empty matrix} as a matrix with zero rows or zero columns (or both). We will assume that an empty matrix is full rank with rank $0$. The product of an empty matrix and another matrix is an empty matrix, e.g., if $\As$ is a $0 \times 3$ empty matrix and $\Bs$ is a $3 \times 5$ matrix, then $\As \Bs$ is an empty matrix of size $0 \times 5$.

%
%

\section{Problem Statement} \label{sec:problemstatement}

We now give a formal problem statement for compute--forward.
Consider the $K$-user discrete memoryless multiple-access channel (DM-MAC)
\begin{equation}
(\Xc_1\times\cdots\times \Xc_K, p_{Y|X_1,\ldots,X_K}, \Yc)
\end{equation}
which consists of $K$ input alphabets $\Xc_k$, $k\in[1:K]$, one receiver alphabet $\Yc$, and a collection of conditional pmfs $p_{Y|X_1,\ldots,X_K}$. See Figure~\ref{fig:probstatement} for an illustration. 

Consider a finite field $\Fq$ and let $\As_1,\ldots,\As_L \in \Fq^K$ denote coefficient vectors. Define
\begin{align}
\As = \begin{bmatrix} \As_1 \\ \vdots \\ \As_L \end{bmatrix} \in \Fq^{L\times K}
\end{align}
as a coefficient matrix, with $L \leq K$.

A $(2^{nR_1},\ldots,2^{nR_K},n; \As)$ code for compute--forward consists of
\begin{itemize}
\item $K$ message sets $[2^{nR_k}]$, $k\in[1:K]$
\item $K$ encoders, where encoder $k$ maps each message $m_k \in [2^{nR_k}]$ to a pair of sequences $(u^n_k, x^n_k)(m_k)\in\Fq^n\times\Xc_k^n$ such that $u^n_k(m_k)$ is \textit{injective},
\item $L$ linear combinations for each message tuple $(m_1,\ldots, m_K)$
\begin{equation*}
    w^n_{\As}(m_1,\ldots,m_K)
    = \begin{bmatrix}
        w^n_{\As_1}(m_1,\ldots,m_K)\\
        \vdots\\
        w^n_{\As_L}(m_1,\ldots,m_K)
    \end{bmatrix}
    = \As \begin{bmatrix}
        u^n_{1}(m_1)\\
        \vdots \\
        u^n_{K}(m_K)
    \end{bmatrix},
\end{equation*}
where additions and multiplications are defined over the vector space $\Fq^n$, and
\item a decoder that assigns estimates $(\hat{w}^n_{\As_1},\ldots,\hat{w}^n_{\As_L}) \in \Fq^n \times \cdots \times \Fq^n$ to each received sequence $y^n \in \Yc^n$.
\end{itemize}

We assume that each message $M_k$ is independently and uniformly drawn from $[2^{nR_k}]$. The average probability of error is defined as 
\begin{align*}
    \pen = \P\big\{(\hat{W}^n_{\As_1},\ldots,\hat{W}^n_{\As_L}) \neq (W^n_{\As_1},\ldots,W^n_{\As_L})\big\}.
\end{align*}
We say that a rate tuple $(R_1,\ldots,R_K)$ is achievable for computing the linear combinations with coefficient matrix $\As$ if there exists
a sequence of $(2^{nR_1},\ldots,2^{nR_K},n; \As)$ codes such that $\lim_{n\rightarrow \infty} \pen = 0$. Overall, the goal is for the receiver to recover the linear combinations
\begin{align}
	w_{\As_\ell}^n(m_1,\ldots, m_K)
	= \sum_{k=1}^K a_{\ell,k} u_k^n(m_k), \qquad \ell\in[1:L], \label{eq:linear-combination}
\end{align}
where $a_{\ell,k}$ is the $(\ell,k)$-th entry of $\As$ and the multiplication and summation operations are over $\Fq$. 

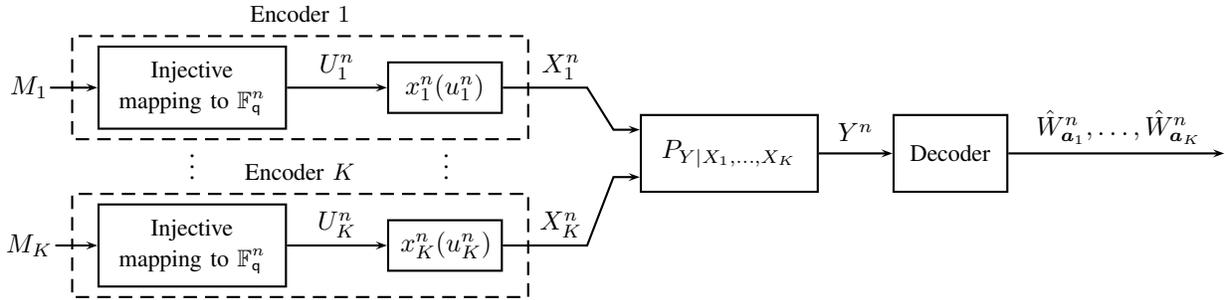
\begin{figure}[ht]
\begin{center}
\psset{unit=0.70mm}
\begin{pspicture}(-42,-15)(189,41)
\rput(-38,25){$M_1$}  \psline{->}(-34,25)(-25,25)
\psframe[linestyle=dashed](-30,15)(57,35) \rput(13.5,39){\small{Encoder $1$}}
\psframe(-25,17)(11,33) \rput(-7,28){{\small{Injective}}} \rput(-7,22){{\small{mapping to $\Fq^n$}}}
\psline{->}(11,25)(30,25) \rput(20.5,29){$U_1^n$}
\psframe(30,20)(52,30) \rput(41,25){{$x_1^n(u_1^n)$}}
\rput(63,29){$X_1^n$} \psline{->}(52,25)(68,25)(73,17)(78,17)

\rput(-7,11.5){$\vdots$}
\rput(41.25,11.5){$\vdots$}

\rput(0,-30){
\rput(-38,25){$M_K$}  \psline{->}(-33,25)(-25,25)
\psframe[linestyle=dashed](-30,15)(57,35) \rput(13.5,39){\small{Encoder $K$}}
\psframe(-25,17)(11,33) \rput(-7,28){{\small{Injective}}} \rput(-7,22){{\small{mapping to $\Fq^n$}}}
\psline{->}(11,25)(30,25) \rput(20.5,29){$U_K^n$}
\psframe(30,20)(52,30) \rput(41.25,25){{$x_K^n(u_K^n)$}}
\rput(63,29){$X_K^n$} \psline{->}(52,25)(68,25)(73,38)(78,38)
}

\psframe(78,5)(112,20) \rput(95,12.5){$P_{Y|X_1,\ldots,X_K}$}
\psline{->}(112,12.5)(126,12.5) \rput(119,16.5){$Y^n$}

\psframe(126,5)(148,20) \rput(137,12.5){\small{Decoder}}
\psline{->}(148,12.5)(189,12.5)\rput(169,17.5){$\hat{W}_{\ab_1}^n,\ldots,\hat{W}_{\ab_K}^n$}

\end{pspicture}
\end{center}
\caption{Block diagram of the compute--forward problem. Each transmitter has a message $M_k$ drawn independently and uniformly from $[2^{nR_k}]$ that is injectively mapped to a representative sequence $U^n_k(M_k)$ over a finite vector space $\Fq^n$, and then into a channel input $X^n_k(M_k) \in \mathcal{X}_k^n$. The $K$ channel inputs pass through a memoryless MAC described by conditional probability distribution $P_{Y|X_1,\ldots,X_K}$ resulting in channel output $Y^n$. Finally, the decoder makes estimates $\hat{W}_{\ab_1}^n,\ldots,\hat{W}_{\ab_K}^n$ of the linear combinations $W_{\ab_\ell}^n(M_1,\ldots,M_K) = \sum_{k} a_{\ell, k} U_k^n(M_k)$.} \label{fig:probstatement}
\end{figure}

\begin{remark}
The role of the mappings $u^n_k(m_k)$ is to embed the messages into the vector space $\Fq^n$, so that it is possible to take linear combinations. The restriction to injective mappings ensures that, given enough linear combinations, it is possible to solve the system of linear equations and recover the original messages (subject to appropriate rank conditions). 
\end{remark}

%
%

\section{Main results}\label{sec:mainresults}

In this section, we present our main results. 
We begin by establishing a joint-decoding-based achievable rate region for computing $L$ linearly independent combinations for a discrete memoryless MAC.

For a coefficient matrix $\Fs \in \Fq^{L_\Fs \times K}$, let us define the notation
\begin{equation}
    W_\Fs = \Fs \begin{bmatrix} U_1 & \cdots & U_K \end{bmatrix}^\tra
\end{equation}
for a vector of linear combinations of $(U_1, \ldots, U_K) \in \Fq^K$. The following theorem establishes our main result on computing $L$ linear combinations.

\begin{theorem}[Compute--forward for the DM-MAC] \label{thm:LK-joint}
A rate tuple $(R_1,\ldots, R_K)$ is achievable for recovering the $L$ linear combinations with coefficient matrix $\As \in \Fq^{L \times K}$ if, for some pmf $\prod_{k=1}^K p(u_k)$ and symbol mappings $x_k(u_k)$, $k\in\Kc$, it is contained in 
\begin{equation}\label{eq:region}
    \mathscr{R}_{\joint}
    = \bigcup_{\Bs} \bigcap_{\Cs} \bigcup_{\Sc} \bigcap_{\Tc} \Bigl\{ (R_1, \dotsc, R_K) \in \mathbb{R}_+^K \colon \textstyle\sum_{k\in\Tc} R_k < H(U(\Tc)) - H(W_{\Bs} | Y, W_{\Cs\Bs}) \Bigr\}
\end{equation}
where the set operations are over all tuples $(\Bs,\Cs,\Sc,\Tc)$ satisfying the following constraints:
\begin{enumerate}
\item $\Bs \in \Fq^{L_\Bs\times K}$ runs over all full-rank matrices such that $1 \leq L_\Bs \leq K$ and $\Span(\Bs) \supseteq \Span(\As)$,
\item $\Cs \in \Fq^{L_\Cs\times L_\Bs}$ runs over all full-rank matrices (including empty matrices) such that $0 \le L_\Cs < L_\Bs$,
\item $\Sc\subseteq [1:L_\Bs]$ runs over all index sets of size $|\Sc|=L_\Bs-L_\Cs$ satisfying
\begin{equation}\label{eq: KL-joint-condition1}
    \rank\left( \begin{bmatrix}
        \Cs \\ 
        \Ib(\Sc)
    \end{bmatrix} \right)
    = L_\Bs,
\end{equation}
\item $\Tc\subseteq\Kc$ runs over all index sets of size $|\Tc|=L_\Bs-L_\Cs$ satisfying
\begin{equation}\label{eq: KL-joint-condition2}
    \rank\left( \begin{bmatrix}
        \Bs(\Sc) \\
        \Ib(\Kc\setminus \Tc)
    \end{bmatrix} \right)
    = K.
\end{equation}
\end{enumerate}
\end{theorem} The coding strategy and error analysis are provided in Section~\ref{sec:proof}. In the following, we give some remarks on Theorem~\ref{thm:LK-joint}.

\begin{remark}
Note that the rate region in curly braces in~\eqref{eq:region} is a function of $\Bs$, $\Cs$ and $\Tc$ but not of $\Sc$. Rather, in the context of the intersection $\cap_\Tc$, $\Tc$ runs over a set that depends on $\Sc$ (cf.~\eqref{eq: KL-joint-condition2}).
\end{remark}

\begin{remark}
By the Steinitz Lemma~\cite{Katznelson--Katznelson2008}, there always exists at least one $\Sc\subseteq\Kc$ such that \eqref{eq: KL-joint-condition2} is satisfied.
\end{remark}

\begin{remark}
Without loss of generality, in the evaluation of $\mathscr{R}_{\joint}$ we can restrain $\Cs$ to being in {\em reduced row echelon form}~\cite{Meyer2000} since the right-hand side of~\eqref{eq:region} only depends on $\Cs$ via $\Span(\Cs)$. Equivalently, we have that $I(U_\Sc; Y, W_\Cs) =I(U_\Sc; Y, W_{\Cs'})$ for any $\Cs, \Cs'$ such that $\text{span}(\Cs)=\text{span}(\Cs')$ since $\Cs$ and $\Cs'$ are deterministic functions of one another. This simplification can be applied for any of the corollaries of Theorem~\ref{thm:LK-joint} that follow.
\end{remark}

Theorem~\ref{thm:LK-joint} admits a direct generalization to multiple receivers. For instance, assume there are $K$ transmitters that communicate with $N$ receivers across the discrete memoryless channel $p_{Y_1,\ldots,Y_N|X_1,\ldots,X_K}$ and that the $i^{\text{th}}$ receiver observes channel output $Y_i$ and wants the linear combinations with coefficient matrix $\As^{(i)}$. Let $\mathscr{R}_{\joint}^{(i)}$ denote~\eqref{eq:region} evaluated with $\As^{(i)}$ in place of $\As$ and $Y_i$ in place of $Y$. Then, a rate tuple $(R_1,\ldots,R_K)$ is achievable if, for some pmf $\prod_{k=1}^K p(u_k)$ and symbol mappings $x_k(u_k)$, $k\in\Kc$, it is contained in $\bigcap_{i = 1}^N \mathscr{R}_{\joint}^{(i)}$.

The rate region in Theorem~\ref{thm:LK-joint} can be easily extended to include a time-sharing random variable using standard arguments~\cite{El-Gamal--Kim2011}. Note that if there are multiple receivers, then the intersection over rate regions should be taken before the convexification due to time-sharing.

The following corollary simplifies Theorem~\ref{thm:LK-joint} for computing one linear combination over a two-user DM-MAC, i.e., $\As\in\Fq^{1\times 2}$ and $K=2$.
In particular, we consider the cases with $\As=[a_1\,\, a_2]$ where $a_1\neq 0$ and $a_2\neq 0$ to avoid degenerate cases. The case when $\rank(\As)=2$ and $K=2$ will be considered afterwards.

\begin{corollary}[Two users, one linear combination]\label{cor:2-user-one-compute}
Consider the case with $K=2$ and $L=1$. A rate pair $(R_1, R_2)$ is achievable for computing one linear combination with respect to the coefficients $\As=[a_1\,\, a_2]$ over a two-user DM-MAC if 
\begin{align}
(R_1, R_2)\in (\mathscr{R}_{\CF} \cup \mathscr{R}_{\LMAC})
\end{align}
for some pmf $p(u_1)p(u_2)$ and symbol mappings $x_1(u_1),x_2(u_2)$, where 
\begin{IEEEeqnarray}{rClrCl}
    \mathscr{R}_{\CF} &=&& \Bigl\{ (R_1, R_2) \colon R_1 &<& H(U_1)-H(W_{\Ab}|Y) \nn\\
        &&& R_2 &<& H(U_2)-H(W_{\Ab}|Y) \Bigr\}, \label{eq:Rcf} \\
    \mathscr{R}_{\LMAC} &=& \rlap{$\,(\mathscr{R}_1\cup\mathscr{R}_2),$} &&&   \label{eq:RLMAC} \\
    \mathscr{R}_k &=&& \Bigl\{ (R_1, R_2) \colon R_1 &<& I(X_1; Y| X_2) \nn\\
        &&& R_2 &<& I(X_2; Y| X_1) \nn\\
        &&& R_1+R_2 &<& I(X_1, X_2; Y) \nn\\
        &&& R_k &<& \min_{\Cs\in\Fqh^{1\times 2}} I(U_k;Y, W_{\Cs}) \Bigr\},
\end{IEEEeqnarray}
and $\Fqh = \Fq\setminus\{0\}$.
\end{corollary}
The proof of Corollary~\ref{cor:2-user-one-compute} is deferred to Appendix~\ref{app:proof_cor1}.

In the following, we explain how our DMC results are related to the lattice compute--forward strategy by Nazer and Gastpar~\cite{Nazer--Gastpar2011} by specializing Corollary~\ref{cor:2-user-one-compute} to the two-user Gaussian MAC, 
\begin{align}\label{eq:Gaussian_MAC}
    Y^n= \hb\left[\begin{array}{c}
    x^n_1\\
    x^n_2
    \end{array}
    \right]+Z^n
\end{align}
where $\hb=[h_{1} ~h_2]$ is the vector of channel gains, the noise $Z^n$ is i.i.d.~$\mathcal{N}(0,1)$, and the channel inputs are subject to average power constraints $\sum_{i=1}^n x^2_{ki} \le nP_k$. The goal is to recover the linear combination with \emph{integer}\footnote{It can be shown that, if the channel coefficients and power constraints are bounded, then we can select a large enough finite field such that any integer-linear combination of codewords (with a positive sum rate) has a corresponding finite field combination. Thus, we can evaluate the rate region by solving a special case of the shortest vector problem, which can be efficiently solved for $K = 2$ by Gauss' algorithm as well as for $K > 2$ by the algorithm proposed in~\cite{Sahraei--Gastpar2017}.} coefficient vector $\As = [a_1 ~a_2] \in \Integer^{1 \times 2}$ again assuming that $a_1 \neq 0$ and $a_2 \neq 0$ to avoid degenerate cases.
In~\cite{Lim--Feng--Pastore--Nazer--Gastpar2018}, we have shown via a discretization method that the rate region $\mathscr{R}_{\CF}$ in Corollary~\ref{cor:2-user-one-compute} can be specialized to the Gaussian case in the form of
\begin{IEEEeqnarray}{rClrCl}
    \mathscr{R}_{\CF} &=&& \Bigl\{ (R_1, R_2) \colon R_1 &<& h(U_1)-h(W_{\As}|Y) + \log\gcd({\As}), \nn\\
        &&& R_2 &<& h(U_2)-h(W_{\As}|Y)+ \log\gcd({\As}) \Bigr\}, \label{eq:Rcf-gaussian} 
\end{IEEEeqnarray}
where $U_k\sim \mathcal{N}(0, P_k)$, the symbol mappings are $X_k = U_k$, and $\gcd{(\As)}$ is the greatest common divisor of $|a_1|$ and $|a_2|$. Specifically, the inequalities in~\eqref{eq:Rcf-gaussian} 
\begin{subequations}\label{eq:Gaussian-CF} evaluate to
\begin{align}
    R_1 &< \frac{1}{2} \log\left(\frac{P_1}{\As\left(\Sigma^{-1}+\hb^{\sf T}\hb\right)^{-1}\As^{\sf T}}\right) + \log\gcd{(\As)},\label{eq:Gaussian-CF1}\\
    R_2 &< \frac{1}{2} \log\left(\frac{P_2}{\As\left(\Sigma^{-1}+\hb^{\sf T}\hb\right)^{-1}\As^{\sf T}}\right)+ \log\gcd{(\As)},\label{eq:Gaussian-CF2}
\end{align}    
\end{subequations}
where $\Sigma=\diag(P_1, P_2)$. The rate region given by the inequalities in~\eqref{eq:Gaussian-CF} is the compute--forward rate region for asymmetric powers from~\cite{Nazer--Cadambe--Ntranos--Caire2016}. Thus, the joint typicality approach can recover the best-known achievable rate region based on nested lattice codes.

Let $\mathscr{R}^*_{\CF}$ denote the rate region from~\eqref{eq:Rcf} evaluated with respect to a choice of $\ab \in \Fq^{1\times 2}$ that minimizes $H(W_{\ab}|Y)$. Let
\begin{IEEEeqnarray}{rCl}
    \mathscr{R}_{\MAC}=\Bigl\{(R_1, R_2) \colon R_1 &<& I(X_1; Y|X_2),\nn\\
    R_2 &<& I(X_2; Y|X_1),\nn\\
    R_1+R_2 &<& I(X_1, X_2; Y)\Bigr\}   \label{eq:RMAC}
\end{IEEEeqnarray} denote the two-user multiple-access achievable rate region (for a fixed input distribution and without time sharing). As shown in~\cite[App.~E]{Lim--Feng--Pastore--Nazer--Gastpar2018}, $\mathscr{R}_{\MAC} \subseteq \mathscr{R}_{\CF}^* \cup \mathscr{R}_{\LMAC}$ since $\mathscr{R}_{\CF}^*$ fills in the defect in $\mathscr{R}_{\LMAC}$. Since the relation holds for both regions without time sharing, the inclusion relation obviously extends to the time-sharing case.
An illustration of the rate regions is given in Fig.~\ref{fig: improvement}.

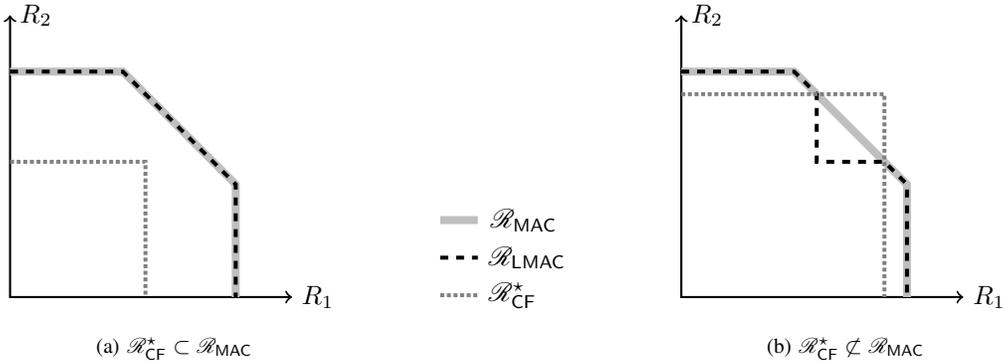
\begin{figure}[ht]
\centering
\subfloat[$\mathscr{R}^\star_{\CF} \subset \mathscr{R}_{\MAC}$]{
\begin{tikzpicture}[scale=1.5,every path/.append style={thick}]
\draw[<->] (0,2.5) node[right] {$R_2$} |- (2.5,0) node[right] {$R_1$};
\draw[line width=3.0pt,color=gray!50!white] (0,2) -- (1,2) -- (2,1) -- (2,0);
\draw[line width=1.5pt,dashed] (0,2) -- (1,2) -- (2,1) -- (2,0);
\draw[line width=1.5pt,densely dotted,color=gray] (0,1.2) -| (1.2,0);
\end{tikzpicture}
\label{fig:R_LMAC_1}
}
\hspace{1cm}
\begin{tikzpicture}[every path/.append style={thick}]
\draw[line width=3.0pt,color=gray!50!white] (0,0) --+ (5mm,0) node[right] {\textcolor{black}{$\mathscr{R}_{\MAC}$}};
\draw[line width=1.5pt,dashed] (0,-5mm) --+ (5mm,0) node[right] {\textcolor{black}{$\mathscr{R}_{\LMAC}$}};
\draw[line width=1.5pt,densely dotted,color=gray] (0,-1cm) --+ (5mm,0) node[right] {\textcolor{black}{$\mathscr{R}^\star_{\CF}$}};
\end{tikzpicture}
\hspace{1cm}
\subfloat[$\mathscr{R}^\star_{\CF} \not\subset \mathscr{R}_{\MAC}$]{
\begin{tikzpicture}[scale=1.5,every path/.append style={thick}]
\draw[<->] (0,2.5) node[right] {$R_2$} |- (2.5,0) node[right] {$R_1$};
\draw[line width=3.0pt,color=gray!50!white] (0,2) -- (1,2) -- (2,1) -- (2,0);
\draw[line width = 1.5pt,dashed] (0,2) -- (1,2) -- (1.2,1.8) |- (1.8,1.2) -- (2,1) -- (2,0);
\draw[line width = 1.5pt,densely dotted,color=gray] (0,1.8) -| (1.8,0);
\end{tikzpicture}
\label{fig:R_LMAC_2}
}
\caption{
An illustration of $\mathscr{R}^\star_{\CF}$, which is the $\mathscr{R}_{\CF}$ rate region~\eqref{eq:Rcf} evaluated with respect to a coefficient vector $\ab$ that minimizes $H(W_{\ab}|Y)$. For the two-user rate region $\mathscr{R}_{\LMAC} = \mathscr{R}_1 \cup \mathscr{R}_2$ in Corollary~\ref{cor:2-user-one-compute}, if $\mathscr{R}^\star_{\CF}$ is contained in $\mathscr{R}_{\MAC}$ as in \protect\subref{fig:R_LMAC_1}, then $\mathscr{R}_{\LMAC}$ and $\mathscr{R}_{\MAC}$ coincide. Otherwise, if $\mathscr{R}^\star_{\CF}$ protrudes out from $\mathscr{R}_{\MAC}$ as in \protect\subref{fig:R_LMAC_2}, then $\mathscr{R}_{\LMAC}$ is obtained by mirroring the protruding part along the dominant face, and removing it from $\mathscr{R}_{\MAC}$. }\label{fig: improvement}
\end{figure}

For the special case of $\As=\Ib$, the computation problem reduces to the conventional multiple-access problem, that is, we recover all $K$ messages individually. In the following corollary, we specialize Theorem~\ref{thm:LK-joint} by fixing $\Bs=\Ib$ for the multiple-access case. Note that, since our proposed coding scheme is constrained by the use of nested linear codes, our achievable rate region does not always match the multiple-access capacity region.

\begin{corollary}[Multiple access via nested linear codes] \label{cor: KK-joint}
A rate tuple $(R_1,\ldots, R_K)$ is achievable for multiple access with nested linear codes if there exists some pmf $\prod_{k=1}^Kp(u_k)$ and symbol mappings $x_k(u_k)$, $k\in\Kc$ such that, for each natural number $0 \le L_\Cs < K$ and each full-rank matrix $\Cs \in \Fq^{L_\Cs \times K}$, we can select a subset $\Sc\subseteq \Kc$ (that can depend on $\Cs$) of size $|\Sc|=K- L_\Cs$ satisfying
\begin{align}
R(\Sc) &< I(U(\Sc); Y, W_\Cs),
\end{align}
and
\begin{align}\label{eq: KK-joint-condition}
\rank\left( \left[ \begin{array}{c}
\Cs\\ \Is(\Sc)
\end{array}
\right]\right)=K.
\end{align}
\end{corollary}
\medskip

Corollary~\ref{cor: KK-joint} is immediate from Theorem~\ref{thm:LK-joint} by setting $L=K$, $\As=\Bs=\Ib$. Moreover, for each $\Sc$ we only have $\Tc=\Sc$ which satisfies~\eqref{eq: KL-joint-condition2}.

\begin{corollary}\label{cor: 2-user-joint}
A rate pair is achievable for the DM-MAC via nested linear codes if $(R_1, R_2)\in \mathscr{R}_{\LMAC}$ for some pmf $p(u_1)p(u_2)$ and symbol mappings $x_1(u_1),x_2(u_2)$, where $\mathscr{R}_{\LMAC}$ is defined in~\eqref{eq:RLMAC}.
\end{corollary} This follows directly from the $\mathscr{R}_{\LMAC}$ evaluation from the proof of Corollary~\ref{cor:2-user-one-compute} in Appendix~\ref{app:proof_cor1}.

\begin{remark}\label{rmk: improvement}
Let ${\mathscr{R}}_{\LMAC,\mathrm{old}}$ denote the rate region in~\cite[Theorem~5]{Lim--Feng--Pastore--Nazer--Gastpar2018} and recall the region $\mathscr{R}_{\MAC}$ in~\eqref{eq:RMAC}. For a fixed distribution $p(u_1)p(u_2)$ and symbol mappings $x_1(u_1)$, $x_2(u_2)$, if $\mathscr{R}_{\CF}$ is strictly contained in $\mathscr{R}_{\MAC}$, then $\mathscr{R}_{\LMAC,\mathrm{old}}$ is strictly contained in $\mathscr{R}_{\MAC}$ whereas $\mathscr{R}_{\LMAC}$ is equal to $\mathscr{R}_{\MAC}$. Thus,
Corollary~\ref{cor: 2-user-joint} strictly improves upon our previous results~\cite[Theorem~5]{Lim--Feng--Pastore--Nazer--Gastpar2018} for the two-user case, i.e., $\mathscr{R}_{\LMAC,\mathrm{old}}$ is contained in $\mathscr{R}_{\LMAC}$. 
\end{remark}

In general, simultaneous decoding offers better performance than sequential decoding. However, for some applications a sequential decoder may offer a better compromise by lowering the implementation complexity, perhaps at the expense of rate (cf.~successive cancellation decoding vs.~joint decoding for multiple access).
For notational convenience, let $\As_k$ denote $\As(\{k\})$ and $\As^{k}=\As(\{1,\ldots, k\})$.
Extending the basic idea of successive cancellation to the computation problem, a decoder could first recover the linear combination corresponding to the first row of $\As$, i.e., $W^n_{\As_1}$, then use the channel output and the linear combination pair $(W^n_{\As_1}, Y^n)$ to recover a second linear combination corresponding to the second row $\As_2$ and so on (see Figure~\ref{fig:sequential-decoding}). Based on this sequential decoding strategy, the following theorem establishes a sequential decoding rate region for computing multiple linear combinations.

\begin{figure}[ht!]
\centering
\small
\psfrag{x1}[c]{$Y^n$}
\psfrag{vd}[c]{$\vdots$}
\psfrag{d1}[c]{Decoder 1}
\psfrag{d2}[c]{Decoder 2}
\psfrag{d3}[c]{Decoder $L$}
\psfrag{d4}[c]{Decoder}
\psfrag{w1}[c]{$\qquad\Wh^n_{\As_1}$}
\psfrag{w2}[c]{$\qquad(\Wh^n_{\As_1}, \Wh^n_{\As_2})$}
\psfrag{w3}[c]{$\qquad(\Wh^n_{\As_1},\ldots, \Wh^n_{\As_L})$}
\includegraphics[scale=0.75]{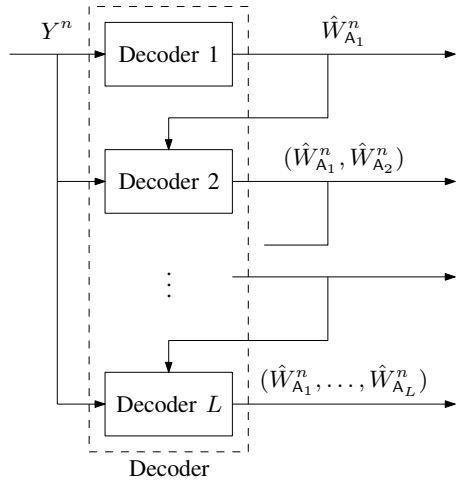}
\caption{Sequential decoder for recovering multiple linear combinations.}
\label{fig:sequential-decoding}
\end{figure}

\begin{theorem}[Sequential decoding]\label{thm:succ-decoding}
A rate tuple $(R_1,\ldots, R_K)$ is achievable for  computing the linear combinations with coefficient matrix $\As \in \Fq^{L \times K}$ if, for some pmf $\prod_{k=1}^Kp(u_k)$, symbol mappings $x_k(u_k)$, $k\in\Kc$ and full-rank matrix $\Bs\in\Fq^{L_\Bs\times K}$, $L \leq L_{\Bs} \leq K$ satisfying $\Span(\As)\subseteq \Span(\Bs)$, we have that
\begin{align}
R_k &< H(U_k)-H(W_{\Bs_j}|Y, W_{\Bs^{j-1}}),\label{eq:SD}
\end{align} for all $1\le j\le L_\Bs$ and $k\in \Kc(\Bs_j)$
where $\Bs_j$ is the $j$-th row of $\Bs$ and $\Bs^j=\Bs([1:j])$, and $\Kc(\Bs_j)=\{k\in\Kc: \Bs_{jk}\neq 0\}$.
\end{theorem}

\begin{IEEEproof}
Consider the case of recovering a single linear combination ($L=1$) corresponding to a vector $\tilde \As \in \Fq^{1 \times K}$. For this case, we evaluate Theorem~\ref{thm:LK-joint} by fixing $\As=\Bs=\tilde \As$. The resulting region is the set of rates $(R_1,\ldots, R_K)$ such that for $k\in\Kc(\tilde \As)$,
\begin{align*}
    R_k &< H(U_k)-H(W_{\tilde \As}|Y),
\end{align*}
for some pmf $\prod_{k=1}^Kp(u_k)$, symbol mappings $x_k(u_k)$, $k\in\Kc$. 
Theorem~\ref{thm:succ-decoding} then follows from Theorem~\ref{thm:LK-joint} upon replacing $Y$ with $(Y, W_{\Bs^{j-1}})$ (i.e., by including $W_{\Bs^{j-1}}$ as an additional channel output at step $j$) and replacing $\tilde \As$ with $\Bs_j$.
\end{IEEEproof}

Let
\begin{IEEEeqnarray}{rCl}
    \mathscr{R}_\joint(\Bs)
    &=& \bigcap_{\Cs} \bigcup_{\Sc} \bigcap_{\Tc} \Bigl\{ (R_1, \dotsc, R_K) \in \mathbb{R}_+^K \colon \textstyle\sum\limits_{k\in\Tc} R_k < H(U(\Tc)) - H(W_{\Bs} | Y, W_{\Cs\Bs}) \Bigr\}   \label{Rjoint} \\
    \mathscr{R}_\seq(\Bs)
    &=& \bigcap_{\substack{(j,k) \colon \\ \Bs_{j,k} \neq 0}} \Bigl\{ (R_1, \dotsc, R_K) \in \mathbb{R}_+^K \colon R_k < H(U_k) - H(W_{\Bs_j} | Y, W_{\Bs^{j-1}}) \Bigr\}   \label{Rseq}
\end{IEEEeqnarray}
denote the partial rate regions involved in Theorems~\ref{thm:LK-joint} and~\ref{thm:succ-decoding}, respectively, prior to computing the union over all matrices $\Bs$ satisfying $\Span(\Bs) \supseteq \Span(\As)$. All set operations (unions and intersections) in~\eqref{Rjoint}--\eqref{Rseq} are to be taken over the sets specified in the statements of Theorems~\ref{thm:LK-joint} and~\ref{thm:succ-decoding}, respectively. Thus, we have that
\begin{align}
    \mathscr{R}_{\joint}=\bigcup_\Bs \mathscr{R}_{\joint}(\Bs)
\end{align} and we similarly define
\begin{align}
    \mathscr{R}_{\seq}=\bigcup_\Bs \mathscr{R}_{\seq}(\Bs).
\end{align}

\begin{theorem}\label{thm:SDinJD}
For any $\Bs$, it holds that
\begin{equation}
    \mathscr{R}_\seq(\Bs)\subseteq\mathscr{R}_\joint(\Bs).
\end{equation}
In particular, it follows that $\mathscr{R}_\seq \subseteq \mathscr{R}_\joint$. 
\end{theorem}
The proof of Theorem~\ref{thm:SDinJD} is given in Appendix~\ref{app:proof_cor2}.

\begin{example}\label{ex:dm-mac}
Consider a $K=3$ user DM-MAC with
\begin{align}\label{eq:dmc_example}
Y=\left[\sum_{k=1}^3X_k + Z\right]\bmod{4},
\end{align}
where $\Xc_k=\{0,1\}$, $\Yc=\Zc=\{0,1,2,3\}$, and $Z$ is an additive random noise generated with pmf $p_Z(0)=1-p$ and $p_Z(1)=p_Z(2)=p_Z(3)=p/3$.

\begin{figure}[ht!]
\begin{center}
\footnotesize
\psfrag{n1}[c]{$0$}
\psfrag{n2}[c]{$1$}
\psfrag{n3}[c]{$2$}
\psfrag{n4}[c]{$3$}
\psfrag{y1}[c]{$Y$}
\psfrag{v1}[c]{}
\psfrag{pm}[c]{$1-p$}
\psfrag{p1}[c]{$p/3$}
\psfrag{c1}[c]{}
\psfrag{c2}[c]{}
\psfrag{x1}[c]{$X_1$}
\psfrag{x2}[c]{$X_2$}
\psfrag{x3}[c]{$X_3$}
\includegraphics[scale=0.75]{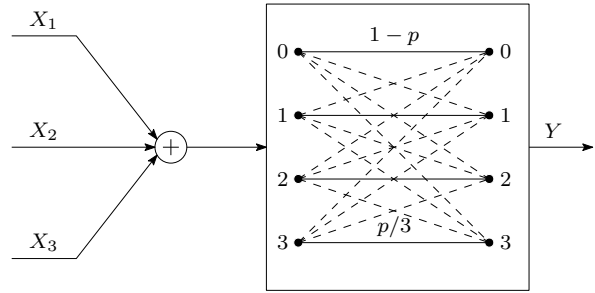}
\caption{Discrete memoryless MAC from Example~\ref{ex:dm-mac}.}
\end{center}
\end{figure}

Figure~\ref{fig:RR_Exp} depicts an inner bound on the joint decoding rate region $\mathscr{R}_\joint $ for the channel from Example~\ref{ex:dm-mac}. This bound is obtained by taking the union over one rank-1 matrix $\Bs=[1,\,1,\,1]$, one rank-3 matrix $\Bs=\Is$, and three rank-2 matrices,
\begin{equation}
\Bs\in\left\{    \begin{bmatrix}
        1 & 0 & 0 \\
        0 & 1 & 1
    \end{bmatrix},
        \begin{bmatrix}
        0 & 1 & 0 \\
        1 & 0 & 1
    \end{bmatrix},
    \begin{bmatrix}
        0 & 0 & 1 \\
        1 & 1 & 0
    \end{bmatrix}\right\}.
\end{equation}
Since this union does not exhaust all possibilities for the values of $\Bs$, it might fall short of yielding the full rate region specified by Theorem~\ref{thm:LK-joint}.

The sequential decoding points in Theorem~\ref{thm:succ-decoding} with $\Bs=[1,\,1,\,1]$ and 
\begin{equation}   
    \Bs = 
    \begin{bmatrix}
        1 & 0 & 0 \\     
        0 & 1 & 1 
    \end{bmatrix},
\end{equation} 
are marked as $a$ and $b$, respectively. The region connecting the corner points $c$ is the multiple-access capacity (recovering the messages separately) of the channel. 

\begin{figure}
\begin{center}
\footnotesize
\psfrag{a}[c]{$R_1$}
\psfrag{b}[c]{$R_2$}
\psfrag{c}[cb]{$R_3$}
\psfrag{s1}[cb]{$a$}
\psfrag{s2}[cb]{$b$}
\psfrag{s3}[cb]{$c$}
\includegraphics[scale=0.3]{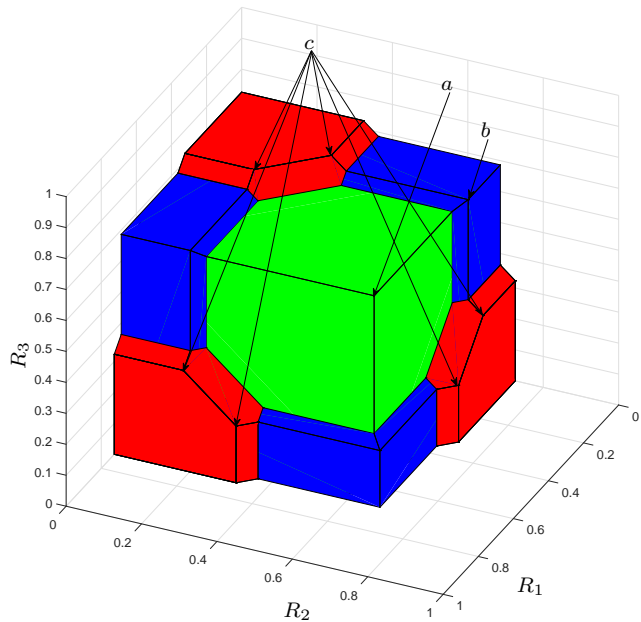}
\end{center}
\caption{An inner bound on the joint decoding rate region for computing $\As=[1,\,1,\,1]$ for the DM-MAC from Example~\ref{ex:dm-mac}.}\label{fig:RR_Exp}
\end{figure}
\end{example}

%
%

\section{Proof of Theorem~\ref{thm:LK-joint}} \label{sec:proof}

We begin by specifying the nested linear codes that will be used as our encoding functions in this paper, starting with some definitions.
For compatibility with linear codes, we define the $\q$-ary expansion of the messages $m_k\in[2^{nR_k}]$ by $\mv_k\in\Fq^{\kappa_k}$, where $\kappa_k=nR_k/\log(\q)$. In addition to the messages, we use auxiliary indices $l_k\in[2^{n\Rh_k}]$, $k=1,\ldots, K$, and similarly define their $\q$-ary expansion by $\lv_k\in\Fq^{\hat\kappa_k}$, where $\hat\kappa_k=n\Rh_k/\log(\q)$. We define
$\Rt_k:=R_k+\Rh_k$, $\Rmax:=\max\{R_1, R_2,\ldots, R_K\}$ and $\Rtmax:=\max\{\Rt_1, \Rt_2, \ldots, \Rt_K\}$.

For notational convenience, we assume that $nR_k/\log(\q)$ and $n\Rh_k/\log(\q)$ are integers for all rates in the sequel.
Further define
\begin{align*}
\mub_k(m_k, l_k) =[\mv_k, \lv_k, \mathbf{0}], \quad k\in\Kc,
\end{align*}
where $\mub_k(m_k, l_k)\in\Fq^\kappa$, $\kappa=n\Rtmax/\log(\q)$, and $\mathbf{0}$ is a vector of zeros with length $n(\Rtmax-\Rt_k)/\log(\q)$.
Note that all $\mub_k(m_k, l_k)$ have the same length due to zero padding. When it is clear from the context, we will simply write $\mub_k$ in place of $\mub_k(m_k, l_k)$. Moreover, since the set $[2^{nR_k}]$ has a one-to-one correspondence to $\Fq^{\kappa_k}$, with some abuse of notation and for simplicity, we will often denote $\mub_k$ as a member of the set $[2^{n\Rt_k}]$, i.e., $\mub_k\in[2^{n\tilde R_k}]$.
\smallskip

We define a $(2^{nR_1},\ldots, 2^{nR_K},2^{n\Rh_1},\ldots, 2^{n\Rh_K},\Fq, n)$ {\em nested linear code} as the collection of $K$ codebooks generated by the following procedure.
\smallskip

\noindent{\bf Codebook generation.}
Fix a finite field $\Fq$ and a parameter $\epsilon' \in (0,1)$.
Randomly generate a $\kappa\times n$ matrix, $\Gb\in\Fq^{\kappa\times n}$, and sequences ${\sf d}^n_k\in\Fq^n$, $k=1,\ldots, K$ where each element of $\Gb$ and ${\sf d}_k^n$ are  randomly and independently generated according to $\U(\Fq)$.

For each $k\in\Kc$, generate a linear code $\Cc_k$ with parameters $(R_k, \Rh_k, n, \q)$ by
\begin{align}
u_k^n(m_k,l_k) &=u_k^n(\mub_k(m_k,l_k)) = \mub_k(m_k, l_k) \Gb \oplus {\sf d}_k^n,\label{eq:codeword}
\end{align}
for $m_k\in[2^{nR_k}]$, $l_k\in[2^{n\Rh_k}]$. Since we have a one-to-one correspondence between $(m_k, l_k)$ and $\mub_k$, we will frequently use $u_k^n(\mub_k)$ to denote $u_k^n(m_k, l_k)$.

As an alternative representation, we write the codebook construction in~\eqref{eq:codeword} by
\begin{align}
\left[\begin{array}{c} u^n_1(\mub_1) \\ \vdots \\ u^n_K(\mub_K)\end{array}\right]= \Mub \Gb \oplus {\sf D},
\end{align}
where 
\begin{align}
\Mub=\left[\begin{array}{c}\mub_1 \\ \vdots \\\mub_K\end{array}\right] \text{ and } {\sf D}=\left[\begin{array}{c}{\sf d}^n_1 \\ \vdots \\{\sf d}^n_K\end{array}\right].
\end{align}
Throughout the proof, we will be interested in the linear dependency between $\mub_1,\ldots, \mub_K$, and representation of messages in matrix form $\Mub$ will be useful. 
Note that from this construction, each codeword is i.i.d.\ uniformly distributed (i.e., $\prod_{i=1}^n p_\q(u_{ki})$ where $p_\q=\U(\Fq)$), and the codewords are pairwise independent.

\medskip
\noindent{\bf Encoding.}
Fix an arbitrary pmf $\prod_{k=1}^Kp(u_k)$, and functions $x_k(u_k)$, $k\in\Kc$.
For $k\in\Kc$, given $m_k\in[2^{nR_k}]$, find an index $l_k\in[2^{n\Rh_k}]$ such that $ u_k^n(m_k,l_k)\in\aepvar(U_k)$. If there is more than one, select one randomly and uniformly. If there is none, randomly choose an index from $[2^{n\Rh_k}]$. Node $k$ transmits $x_{ki}(u_{ki})$, $i=1,\ldots, n$.

Define the collection of $\Mub$ matrices by the set $\Ic$ and the sumset of $\Mub\in\Ic$ with respect to the coefficient matrix $\As$ by
\begin{align*}
\Ic_\text{sumset}(\As) = \{\Mub_\As: \Mub_\As=\As\Mub, \Mub\in\Ic\}.
\end{align*}

\noindent{\bf Decoding.}
Let $\e'<\e$. Upon receiving $y^n$, the decoder searches for a unique index tuple $\tilde \Mub_\As\in\Ic_\text{sumset}(\As)$ such that
$\tilde\Mub_\As=\As\tilde\Mub$ and
\begin{align}
	(u_1^n(\tilde\mub_1), \ldots, u_K^n(\tilde\mub_K), y^n)\in\aep,
\end{align} 
for some $\tilde\Mub\in\Ic$. If it finds a unique index tuple, it declares 
\begin{align}
     \Wh^n_{\As}=\tilde\Mub_{\As}\mathsf{G} \oplus \As \Ds
\end{align}
as its estimate.
Otherwise, if there is no such index tuple, or more than one, the decoder declares an error.

\medskip

\noindent{\bf Analysis of the probability of error.}
Let $M_1, \ldots, M_K$ be the messages, and $L_1, \ldots, L_K$ be the indices chosen by the encoders.
With some abuse of notation, denote by the random variable $\Mub^{\star}_\As$ the true sum of the indices $\mub_k(M_k, L_k)$, $k\in\Kc$ with respect to the coefficients $\As$.
Then, the decoder makes an error only if one or more of the following events occur,
\begin{align*}
\Ec_1&=\{U_k^n(m_k, l_k) \not\in\aepvar \text{ for all } l_k \text{ for some } m_k, k\in\Kc\},\\
\Ec_2&=\{(U_1^n(M_1, L_1), \ldots, U_K^n(M_K, L_K), Y^n)\not\in\aep\},\\
\Ec_3&=\{(U_1^n(m_1, l_1), \ldots, U_K^n(m_K, l_K), Y^n) \in \aep \text{ for some } \Mub\in\Ic \text{ such that } \As\Mub\neq \Mub^{\star}_\As\}.
\end{align*}
Then, by the union of events bound,
\begin{align}
\P(\Ec) &\le \P(\Ec_1)+\P(\Ec_2\cap\Ec_1^c)+\P(\Ec_3\cap\Ec_1^c).
\end{align}
By the covering lemma in~\cite[Lemma 9]{Lim--Feng--Pastore--Nazer--Gastpar2018}, the probability $\P(\Ec_1)$ tends to zero as $n\to\infty$ if
\begin{align}\label{eq:cf-cover}
\Rh_k>  D(p_{U_k}\|p_\q)+\d(\e'),\quad k=1,\ldots,K.
\end{align}

Define $\Mc:=\{M_1=0, \ldots, M_K=0, L_1=0, \ldots, L_K=0\}$ as the event where all messages and the chosen auxiliary indices are zero which also implies that $\Mub^{\star}_{\As}=\mathbf{0}$. By the symmetry of the codebook construction and encoding steps, we have that $P(\Ec_2\cap\Ec_1^c)=P(\Ec_2\cap\Ec_1^c|\Mc)$ and $P(\Ec_3\cap\Ec_1^c)=P(\Ec_3\cap\Ec_1^c|\Mc)$.

By the Markov lemma in~\cite[Lemma 12]{Lim--Feng--Pastore--Nazer--Gastpar2018}, the second term $\P(\Ec_2\cap\Ec_1^c|\Mc)$ tends to zero as $n\to\infty$ if~\eqref{eq:cf-cover} is satisfied.
For the third term,
\begin{align}
\P&(\Ec_3\cap\Ec_1^c|\Mc)\nn\\
&=\P\left\{(U^n_1(\mub_1),\ldots, U^n_K(\mub_K), Y^n)\in\aep \text{ for some } \Mub \text{ such that } \As\Mub\neq \mathbf{0}, \Ec_1^c\right| \Mc\}\nn\\
&\stackrel{(a)}{=}\P \bigl\{(U^n_1(\mub_1),\ldots, U^n_K(\mub_K), W^n_\Bs(\Mub_\Bs), Y^n)\in\aep, \Mub_\Bs=\Bs\Mub \nonumber\\
&\qquad \quad \text{ for some } \Mub \text{ such that } \As\Mub\neq \mathbf{0}, \Ec_1^c\big| \Mc\bigr\}\nn\\
&\le \P\left\{(W^n_\Bs(\Mub_\Bs), Y^n)\in\aep, \Mub_\Bs=\Bs\Mub \text{ for some } \Mub \text{ such that } \As\Mub\neq \mathbf{0}, \Ec_1^c| \Mc\right\}\nn\\
&\stackrel{(b)}{\le} \P\left\{(W^n_\Bs(\Mub_\Bs), Y^n)\in\aep, \Mub_\Bs=\Bs\Mub \text{ for some } \Mub \text{ such that } \Bs\Mub\neq \mathbf{0}, \Ec_1^c| \Mc\right\}\nn\\
&= \P\left\{(W^n_\Bs(\Mub_\Bs), Y^n)\in\aep \text{ for some } \Mub_\Bs\in\Ic_\text{sumset}(\Bs) \text{ such that } \Mub_\Bs \neq \mathbf{0}, \Ec_1^c| \Mc\right\}\nn\\
&\le \sum_{\substack{\Mub_\Bs\in\Ic_\text{sumset}(\Bs):\\ \Mub_\Bs\neq \mathbf{0}}}\P\{(W^n_\Bs(\Mub_\Bs), Y^n)\in\aep, \Ec_1^c |\Mc \} \label{eq:e31}
\end{align}
where $W^n_{\Bs}(\Mub_{\Bs})=\Mub_{\Bs}\mathsf{G} \oplus \Bs \Ds$, $\Bs\in\Fq^{L_\Bs\times K}$ is any rank $L_\Bs$ matrix ($L \le L_\Bs \le K$) such that $\Span(\As)\subseteq \Span(\Bs)$, and step $(a)$ follows since $W^n_\Bs(\Mub_\Bs)$ is a deterministic function of $(U^n_1(\mub_1),\ldots, U^n_1(\mub_K))$ and step $(b)$ follows since $\Bs\Mub=\mathbf{0}$ implies $\As\Mub=\mathbf{0}$.

At a high level, the proof steps up to this point are reminiscent of standard coding theorems based on random~i.i.d. code ensembles, e.g. the multiple-access channel coding theorem proof in~\cite{El-Gamal--Kim2011} and the hybrid coding scheme in~\cite{Minero--Lim--Kim2015}, except that we use specialized joint typicality lemmas and a Markov lemma developed specifically for nested linear code ensembles~\cite[Lemma 12]{Lim--Feng--Pastore--Nazer--Gastpar2018}. Ideally, we would like to upper bound the probability term $\P\{(W^n_\Bs(\Mub_\Bs), Y^n)\in\aep, \Ec_1^c |\Mc \}$ independently of $\Mub_\Bs$ using a joint typicality lemma and then upper bound the cardinality of the set $\Ic_\text{sumset}(\Bs)$ to conclude the proof. However, due to the common nested linear codebook, if the competing index tuples are {\em linearly dependent} on the true index tuples, then the competing codewords are statistically dependent with the true codewords. This dependency means that we cannot directly apply the standard packing lemma (e.g. \cite[Lemma 3.1]{El-Gamal--Kim2011}). In the following, we resolve this difficulty by partitioning the sum index tuples $\Mub_\Bs$ and present a joint typicality lemma that can be applied to each subset separately. 

To this end, we proceed with some definitions.
Define the set
\begin{align}
\Lc_\Bs&=\{\Mub_\Bs:  \Mub_\Bs\in\Ic_\text{sumset}(\Bs), \Mub_\Bs\neq \mathbf{0}\}.
\end{align}
We further divide the set $\Lc_\Bs$ into the cover 
\begin{align*}
\Lc_{\Bs}(r, \Cs) &= \{ \Mub_\Bs: \Mub_\Bs\in\Lc_\Bs, \rank(\Mub_\Bs)=r, \Cs\Mub_\Bs=\mathbf{0}\}\\
&= \{ \Mub_\Bs: \Mub_\Bs\in\Lc_\Bs, \nullspace(\Mub_\Bs^\Tr)=\Span(\Cs)\},
\end{align*}
for $1\le r\le L_\Bs$ and $\Cs\in\Fq^{(L_\Bs-r)\times L_\Bs}$. For the case $r=L_\Bs$, $\Cs=\emptyset$,
\begin{align*}
\Lc_{\Bs}(L_\Bs, \emptyset) &= \{ \Mub_\Bs: \Mub_\Bs\in\Lc_\Bs, \rank(\Mub_\Bs)=L_\Bs\}.
\end{align*}

Note that we have $\Lc_\Bs=\cup_r\cup_\Cs\Lc_\Bs(r, \Cs)$.
The set $\Lc_{\Bs}(r, \Cs)$ divides $\Lc_{\Bs}$ into subsets of $\Mub_{\Bs}$ that have the same nullspace, where the nullspaces are represented by a generator matrix $\Cs$ (with rank of $L_{\Bs}-r$).

We are now ready to proceed with the last probability term $P(\Ec_3\cap\Ec_1^c|\Mc)$ using the union of events bound. 
Continuing from~\eqref{eq:e31},
\begin{align*}
P(\Ec_3\cap\Ec_1^c|\Mc)&\le\sum_{\substack{\Mub_\Bs\in\Ic_\text{sumset}(\Bs):\\ \Bs\Mub\neq \mathbf{0}}}\P\{(W^n_\Bs(\Mub_\Bs), Y^n)\in\aep, \Ec_1^c |\Mc \} \\
&\le \sum_{r=1}^{L_\Bs} \sum_{\Cs}\sum_{\Mub_\Bs\in\Lc_\Bs(r,\Cs)}\P\left\{ (W^n_\Bs(\Mub_\Bs), Y^n)\in\aep , \Ec_1^c |\Mc\right\},\\
& \stackrel{(a)}{\le}  \sum_{r=1}^{L_\Bs} \sum_{\Cs} 2^{n\max_{\Tc}(\sum_{k\in\Tc} R_k+\sum_{k\in\Tc} \Rh_k)}  2^{-n(I(W_{\Bs(\Sc)}; Y, W_{\Cs\Bs})+D(p_{W_{\Bs(\Sc)}}\|p_\q^{|\Sc|})+\tilde D-\d(\e))},
\end{align*}
where $\tilde D=\sum_{k\in\Kc} \big(D(p_{U_k}\| p_\q)-\Rh_k\big)$ and $(a)$ holds for any $\Sc\subseteq\Kc$ such that $|\Sc|=r$ and
\begin{equation}
    \rank\left(\begin{bmatrix} \Cs \\ \Is(\Sc) \end{bmatrix}\right) = L_\Bs, \label{eq:Scond}
\end{equation}
the maximum in the exponent of the last inequality is over all subsets $\Tc \in \Kc$ such that $|\Tc|=r$ and 
\begin{equation}
    \rank\left(\begin{bmatrix} \Bs(\Sc) \\ \Is(\Kc\setminus\Tc) \end{bmatrix}\right) = K. \label{eq:Tcond}
\end{equation}
In step $(a)$, we have applied the following two key lemmas that provide a cardinality bound on the set $\Lc_{\Bs}(r, \Cs)$ and a joint typicality lemma for nested linear codes.
The proofs of these lemmas are deferred to Appendix~\ref{app:proof_of_lemmas}.

\begin{lemma}[Cardinality bound]\label{lem:card1}
Let $L_\Bs$ and $r$ be integers such that $L_\Bs\le K$ and $1\le r\le L_\Bs$.
Let $\Bs\in\Fq^{L_\Bs\times K}$, $L_\Bs\le K$ and $\Cs\in\Fq^{(L_\Bs-r)\times L_{\Bs}}$ be full rank matrices.
Then, for any $\Sc\subseteq\Kc$ such that $|\Sc|=r$ and
\begin{align}\label{eq:S-cond}
\rank\left(\left[\begin{array}{c}\Cs\\ \Is(\Sc)\end{array}\right]\right)=L_\Bs, 
\end{align}
we have that
\begin{align}\label{eq:cardinality}
|\Lc_{\Bs}(r, \Cs)| &\le \max_{\Tc}2^{n(\sum_{k\in\Tc} R_k+\sum_{k\in\Tc} \Rh_k)}, 
\end{align}
where the maximization is over all subsets $\Tc \in \Kc$ such that $|\Tc|=r$ and
\begin{align*}
\rank\left(\left[\begin{array}{c}\Bs(\Sc)\\ \Is(\Kc\setminus\Tc)\end{array}\right]\right)=K.
\end{align*}
\end{lemma}

\begin{lemma}[Joint typicality lemma for nested linear codes]\label{lem:jt-lemma}
Consider $1\le r \le L_{\Bs}$ and $\Cs \in\Fq^{(L_\Bs-r)\times L_\Bs}$ such that $\rank(\Cs)=L_\Bs-r$ and assume that $ \Mub_\Bs \in \Lc_\Bs(r, \Cs)$.
Then,
\begin{align*}
\P&\left\{ (W^n_\Bs(\Mub_\Bs), Y^n)\in\aep , \Ec_1^c |\Mc\right\}\le 2^{-n(I(W_{\Bs(\Sc)}; Y, W_{\Cs\Bs})+D(p_{W_{\Bs(\Sc)}}\|p_\q^{|\Sc|})+\tilde{D}-\d(\e))}.
\end{align*}
where $\tilde D=\sum_{k\in\Kc} \big(D(p_{U_k}\| p_\q)-\Rh_k\big)$.
\end{lemma}

Thus, for any $\Bs\in\Fq^{L_{\Bs}\times K}$, $\Span(\As)\subseteq \Span(\Bs)$, we have a bound on $P(\Ec_3\cap\Ec_1^c|\Mc)$ that tends to zero as $n\to\infty$ if 
for all full rank $\Cs\in\Fq^{L_{\Cs}\times L_{\Bs}}$, $0\le L_{\Cs} < L_{\Bs}$, there exists an $\Sc$ that satisfies~\eqref{eq:Scond} and
\begin{align}
    \Rh_k &> D(p_{U_k}\|p_\q)+\d(\e'),\quad k\in\Kc\\
    \sum_{k\in\Tc} R_k+\sum_{k\in\Tc} \Rh_k& < I(W_{\Bs(\Sc)}; Y, W_{\Cs\Bs})+D(p_{W_{\Bs(\Sc)}}\|p_\q^{|\Sc|})+\tilde D-\d(\e),
\end{align}
for all $\Tc$ which satisfies~\eqref{eq:Tcond}. To complete the proof, we eliminate the auxiliary rates $\Rh_k$, $k\in\Kc$ and find that
\begin{align}
    \sum_{k\in\Tc} R_k&< I(W_{\Bs(\Sc)}; Y, W_{\Cs\Bs})+D(p_{W_{\Bs(\Sc)}}\|p_\q^{|\Sc|})-\sum_{k\in\Tc}D(p_{U_k}\|p_\q)-\d'(\e')\\
    &\stackrel{(a)}{=}I(W_{\Bs(\Sc)}; Y,W_{\Cs\Bs})-H(W_{\Bs(\Sc)})+\sum_{k\in\Tc}H(U_k)-\d'(\e')\\
    &=H(U(\Tc))-H(W_{\Bs(\Sc)}| Y,W_{\Cs\Bs})-\d'(\e'),
\end{align}
where step $(a)$ is from the relation $D(p_{W_{\Bs(\Sc)}}\|p_\q^{|\Sc|})=|\Sc|\log(\q)-H(W_{\Bs(\Sc)})$ as well as $D(p_{U_k}\|p_\q)=\log(\q)-H(U_k)$ and the fact that $|\Sc|=|\Tc|$.

%
%

\section{Discussions}\label{sec:discussions}

In this paper, we presented a framework for integrating structured code ensembles into the joint typicality framework. 
As a case, we generalized the compute--forward framework to discrete memoryless networks and established a joint decoding rate region for computing any number of linear combinations of the codewords. Our work provides the foreground for a general theorem on arbitrary networks and flows using nested linear codes. 

In this sense, we view the compute--forward framework, not as a fourth paradigm for relaying, but as a new dimension in code construction for relaying strategies. 
This is a more general perspective that views compute--forward as an ``algebraic'' decode--forward strategy as originally suggested by  Abbas El Gamal in his 2010 ISIT Plenary Talk~\cite{El-Gamal2010}. He also posed several interesting questions on how joint typicality coding strategies can be combined with structured codes. Our framework provides the initial tools to redevelop and explore the coding strategies in network information theory using random nested linear code ensembles in place of random i.i.d.~code ensembles.   

One by-product of our analysis is an achievable rate region for multiple access via nested linear codes. Recent work~\cite{Sen--Kim2018} has further explored this rate region and  shown, through a careful selection of both the finite field and symbol mappings, that it in fact corresponds to the full multiple-access capacity region.

Another important aspect of our framework is the resulting {\em simultaneous joint decoding} rate region for compute--forward. The joint typicality decoder presented in this paper was shown to be optimal with respect to nested linear codes in~\cite{Sen--Lim--Kim2018} for the $K=2, L=1$ case. On the other hand, the sharpest-known analysis for a lattice-based compute--forward strategy relies on suboptimal sequential decoding~\cite{Nazer--Cadambe--Ntranos--Caire2016} due to the technical limitations in analyzing joint decoders for lattice codes~\cite{Ordentlich--Erez2013}.

In an effort to build a unifying compute--forward framework that includes all previously-known achievable rate regions (traditionally obtained via lattice codes), one important question is whether the DMC framework presented herein can be translated to the continuous case and to integer-linear combinations over the {\em real} field. In~\cite[Theorems 4,7]{Lim--Feng--Pastore--Nazer--Gastpar2018}, we have already given a proof  for the special case $K=2,\ L=1$. Our discretization method, which borrows a key result from~\cite{Makkuva--Wu2018}, is more involved than most of the classic discretization approaches for information-theoretic quantities (e.g., \cite{Gray1990_A, Cover--Thomas2006, El-Gamal--Kim2011}). The proof of the general case for arbitrary $K$, $L$, requires yet more steps and will appear in an upcoming publication~\cite{Pastore--Lim--Feng--Nazer--Gastpar2018}.

%
%

\appendices

\section{Proof of Corollary~\ref{cor:2-user-one-compute}}\label{app:proof_cor1}

We particularize Theorem~\ref{thm:LK-joint} by setting $K=2$, $L=1$ and $\As=\ab=[a_1\,\, a_2]\in\Fq^{1\times 2}$ with $a_1\neq 0$ and $a_2\neq 0$. Since in the outermost union operation in~\eqref{eq:region}, $\Bs$ runs over all matrices satisfying $\Span(\Bs) \supset \Span(\As)$, we infer that $\Bs$ must run over all full-rank matrices $\Bs \in \Fq^{2 \times 2}$ as well as over all those $\Bs \in \Fq^{1 \times 2}$ that are scalar multiples of $\ab$ (for which case it suffices to consider $\Bs=\ab$). For the sake of simplifying derivations (at the cost of possibly missing out on a part of the achievable rate region), out of all possible full-rank $2$-by-$2$ matrices $\Bs$, we shall only retain the identity matrix $\Bs = \Is$. In summary, the union operation reduces to taking the union over only two matrices, namely $\Bs = \ab$ and $\Bs = \Is$.

For $\Bs = \ab$, by virtue of the constraints on $\Cs$, $\Sc$, $\Tc$ laid out in Theorem~\ref{thm:LK-joint}, in the set operations of~\eqref{eq:region} the matrix $\Cs$ can only be the $0$-by-$1$ empty matrix, $\Sc$ can only be the singleton set $\{1\}$, and $\Tc$ can be either $\{1\}$ or $\{2\}$. The resulting rate region is $\mathscr{R}_{\CF}$ as defined in~\eqref{eq:Rcf}.

For $\Bs = \Is$, $\Cs$ runs over all full-rank (empty) $0$-by-$2$ and $1$-by-$2$ matrices:
\begin{itemize}
    \item   For $\Cs \in \Fq^{0 \times 2}$, $\Sc$ and $\Tc$ can only be equal to $\Sc = \Tc = \{1,2\}$, hence we obtain the sum-rate bound
    \begin{IEEEeqnarray*}{rCl}
        R_1 + R_2
        &<& H(U_1,U_2) - H(W_{\Bs}|Y) \\
        &=& I(U_1,U_2;Y) \\
        &=& I(X_1,X_2;Y).
    \end{IEEEeqnarray*}
    \item   For $\Cs = [c_1\,\, c_2] \in \Fq^{1 \times 2}$ a non-zero vector, we need to further distinguish three cases:
    \begin{itemize}
        \item   Case $c_1 \neq 0$ and $c_2 = 0$: The index sets can only be equal to $\Sc = \Tc = \{2\}$, hence we obtain the rate bound
                \begin{IEEEeqnarray*}{rCl}
                    R_2 &<& H(U_2) - H(U_1,U_2|Y,U_1) \\
                        &=& I(U_2;Y,U_1) \\
                        &=& I(X_2;Y|X_1).   \IEEEyesnumber
                \end{IEEEeqnarray*}
        \item   Case $c_1 = 0$ and $c_2 \neq 0$: similarly to the previous case, the index sets can only be equal to $\Sc = \Tc = \{1\}$, hence we obtain the rate bound
                \begin{IEEEeqnarray}{rCl}
                    R_1 &<& I(X_1;Y|X_2).
                \end{IEEEeqnarray}
        \item   Case $c_1 \neq 0$ and $c_2 \neq 0$: the index sets can be either $\Sc = \Tc = \{1\}$ or $\Sc = \Tc = \{2\}$. For the former, we obtain 
                \begin{IEEEeqnarray*}{rCl}
                    R_1 &<& H(U_1) - H(U_1,U_2|Y,W_\Cs) \\
                        &=& H(U_1) - H(U_1|Y,W_\Cs) \\
                        &=& I(U_1;Y,W_\Cs)   \IEEEyesnumber\label{LMAC_inequality_1}
                \end{IEEEeqnarray*}
                For the latter, we obtain similarly
                \begin{IEEEeqnarray}{rCl}
                    R_2 &<& I(U_2;Y,W_\Cs).   \label{LMAC_inequality_2}
                \end{IEEEeqnarray}
    \end{itemize}
    The last two rate inequalities~\eqref{LMAC_inequality_1}--\eqref{LMAC_inequality_2} are combined via a logical `or' (due to the union over $\Sc$). Recombining the above three case distinctions on the coefficient pair $(c_1,c_2)$ via a logical `and' (due to union over $\Cs$) yields the rate region $\mathscr{R}_\LMAC$ as defined in~\eqref{eq:RLMAC}. Finally, the union over $\Bs$ yields the final rate region $\mathscr{R}_\CF \cup \mathscr{R}_\LMAC$ and proves Corollary~\ref{cor:2-user-one-compute}.
\end{itemize}

%
%

\section{Proof of Theorem~\ref{thm:SDinJD}}\label{app:proof_cor2}

Let us define $\tilde{\mathscr{R}}_\joint(\Bs)$ as the joint decoding region $\mathscr{R}_\joint(\Bs)$ where, for each $\Cs$, we fix an index set $\Sc^{\star}$ chosen according to Algorithm~\ref{alg:s-matrix}. (To streamline our notation, we do not show the dependence of $\Sc^*$ on $\Cs$ explicitly.) In other words, instead of taking the union over all $\Sc$ in Theorem~\ref{thm:LK-joint}, for each $\Cs$ we fix a set $\Sc=\Sc^\star$, which leads to the relation 
\begin{align*}
\tilde{\mathscr{R}}_\joint(\Bs) \subseteq \mathscr{R}_\joint(\Bs),
\end{align*}
since by following the steps in Algorithm~\ref{alg:s-matrix}, $\Sc^\star$ satisfies
\begin{align}
\rank( [\Cs^{\sf T}, \Is(\Sc^\star)^{\sf T}])&=L_\Bs.
\end{align}

\begin{algorithm}[ht!]
\caption{Algorithm for constructing $\Sc^\star$.}
\begin{algorithmic}[1]
\State $\Sc^\star \gets \emptyset$
\For {$i=1:L_\Bs$}
\If {$\eb_i\not\in \Span([\Is(\Sc^\star)^{\sf T}, \Cs^{\sf T}]^{\sf T})$}
\State $\Sc^\star \gets \Sc^\star \cup \{i\}$
\EndIf
\EndFor
\end{algorithmic}
\label{alg:s-matrix}
\end{algorithm}

In the following, we prove the relation
\begin{align}
    \mathscr{R}_\seq(\Bs) \subseteq \tilde{\mathscr{R}}_\joint(\Bs) \label{eq:relation1}
\end{align}
by showing that the rate region $\mathscr{R}_\seq(\Bs)$ satisfies every inequality in 
$\tilde{\mathscr{R}}_\joint(\Bs)$, namely, the set of inequalities
\begin{align*}
\sum_{k\in\Tc} R_k < H(U(\Tc))-H(W_{\Bs(\Sc^\star)}|Y, W_{\Cs\Bs}),
\end{align*}
for all full rank $\Cs \in \Fq^{L_\Cs\times L_\Bs}$, $\Sc^\star$ chosen by Algorithm~\ref{alg:s-matrix}, 
and all $\Tc\subset\Kc$ such that $|\Tc|=L_\Bs-L_\Cs$ with
\begin{align}
\rank([\Bs(\Sc^\star)^{\sf T}, \Is(\Kc\setminus\Tc)^{\sf  T}])&=K.
\end{align}

\begin{lemma}
Let $\Sc^\star$ be chosen according to Algorithm~\ref{alg:s-matrix}.
Then, for $\Tc$ such that~\eqref{eq: KL-joint-condition2} is satisfied, there exists a one-to-one mapping $\sigma_\Tc: \Sc^\star \to \Tc$ such that for all $j\in\Sc^\star$, $\Bs_{j,\sigma_\Tc(j)}\neq 0$.
\end{lemma}

\begin{proof}
Let
\begin{align*}
\hat\Bs=\left[\begin{array}{c}
\Bs(\Sc^\star) \\
\Is(\Kc\setminus\Tc)
\end{array}\right].
\end{align*} 
Since $\hat\Bs$ and $\Bs(\Sc^\star)$ are full rank, we have that $|\Sc^\star|=|\Tc|$. Next, we define a submatrix $\hat\Bs(\Sc^\star, \Tc)$ which is formed by taking the elements $\hat\Bs_{ij}$, $i\in\Sc^\star$ and $j\in\Tc$. Since $\hat\Bs(\Sc^\star, \Tc)$ is a submatrix of $\hat\Bs$, the existence of a permutation $\hat\sigma_{\Tc}:[1:|\Sc^\star|] \to [1:|\Tc|]$ such that $\hat\Bs_{j, \hat\sigma(j)}(\Sc^\star, \Tc)\neq 0$, $j\in[1:|\Sc^\star|]$, implies the existence of a one-to-one mapping $\sigma_\Tc: \Sc^\star\to\Tc$ such that $\hat{\Bs}_{j,\sigma_\Tc(j)}\neq 0$, and thus equivalently, $\Bs_{j,\sigma_\Tc(j)}\neq 0$ for $j\in\Sc^\star$.

To this end, we will show that there exists such a permutation for $\hat\Bs(\Sc^\star, \Tc)$ by contradiction. From the fact that $\hat\Bs$ is full rank and $\Is(\Kc\setminus\Tc)$ is a collection of standard basis vectors, it is easy to see that the submatrix $\hat\Bs(\Sc^\star, \Tc)$ is a full rank matrix. Since $\hat\Bs(\Sc^\star, \Tc)$ is also a square matrix, it is invertible. Suppose that there does not exist such a permutation for $\hat\Bs(\Sc^\star, \Tc)$. Then, for all possible permutations, $\prod_{j=[1:|\Sc^\star|]}\hat\Bs_{j,\hat\sigma_\Tc(j)}(\Sc^\star, \Tc)=0$. Since this implies that the determinant of $\hat\Bs(\Sc^\star, \Tc)$ is zero, it contradicts the fact that it is invertible.  
\end{proof}

By taking the sum over both sides of the inequalities 
\begin{align*}
R_{\sigma_{\Tc}(j)} &< H(U_{\sigma_{\Tc}(j)})-H(W_{\Bs_j}|Y, W_{\Bs^{j-1}}), \quad j\in\Sc^\star
\end{align*}
which are included in the region $\mathscr{R}_{\text{seq}}(\Bs)$, we have
\begin{align}
\sum_{k\in\Tc} R_k &\stackrel{(a)}{<} H(U(\Tc))-\sum_{j\in\Sc^\star}H(W_{\Bs_j}|Y, W_{\Bs^{j-1}})\\
&\le H(U(\Tc))-\sum_{j\in\Sc^\star}H(W_{\Bs_j}|Y, W_{\Cs\Bs}, W_{\Bs^{j-1}})\\
&\stackrel{(b)}{=} H(U(\Tc))-\sum_{j\in\Sc^\star}H(W_{\Bs_j}|Y, W_{\Cs\Bs}, W_{\Bs(\Sc^\star \cap[1:j-1])})\\
&= H(U(\Tc))-H(W_{\Bs(\Sc^\star)}|Y, W_{\Cs\Bs}),
\end{align}
where step $(a)$ follows since there exists a one-to-one mapping $\sigma: \Sc^\star \to \Tc$ such that for all $j \in \Sc, \Bs_{j, \sigma(j)} \ne 0$ and step $(b)$ follows 
from the fact that
\begin{align*}
\Span\left(\left[\begin{array}{c}\Cs\Bs \\ \Bs^{j}\end{array}\right]\right) = \Span\left(\left[\begin{array}{c}\Cs\Bs \\ \Bs(\Sc^\star\cap[1:j])\end{array}\right]\right)
\end{align*} 
since for $k\not\in\Sc^\star$ where $1\le k \le j$, 
\begin{align*}
\eb_k\in \Span\left(\left[\begin{array}{c}\Cs \\ \Is(\Sc^\star\cap[1:k])\end{array}\right]\right)
\end{align*}
according to Algorithm~\ref{alg:s-matrix}. 

Finally, since the relation holds for an arbitrary $\Cs$, we have shown the relation~\eqref{eq:relation1}.  

%
%

\section{Proof of Lemmas~\ref{lem:card1} and~\ref{lem:jt-lemma}}\label{app:proof_of_lemmas}
\subsection{Proof of Lemma~\ref{lem:card1}}
Recall the definition $\Bs(\Sc)=\Is(\Sc)\Bs$. 
First, we show that $\Bs(\Sc)\Mub$ is full rank. By assumption, $\rank(\Mub_\Bs)=r$, where $r\le L_\Bs$. Thus,
\begin{align*}
\rank\left(\left[\begin{array}{c} \Cs\Bs\Mub \\ \Is(\Sc)\Bs\Mub\end{array}\right]\right) = \rank\left(\left[\begin{array}{c} \mathbf{0} \\ \Is(\Sc)\Bs\Mub\end{array}\right]\right)=r.
\end{align*}
Next define
\begin{align*}
\Lc_{\Bs}(r, \Cs, \Sc) &= \{ \Mub_{\Bs(\Sc)}: \Mub_{\Bs(\Sc)}=\Bs(\Sc)\Mub, \Mub\in \Ic, \rank(\Mub_{\Bs})=r, \Cs\Mub_\Bs=\mathbf{0}\},\\
\bar{\Lc}_{\Bs}(r,\Sc) &= \{ \Mub_{\Bs(\Sc)}: \Mub_{\Bs(\Sc)}=\Bs(\Sc)\Mub, \Mub\in \Ic, \rank(\Mub_{\Bs(\Sc)})=r \}.
\end{align*}
Then,
\begin{align*}
|\Lc_{\Bs}(r, \Cs)| &\stackrel{(a)}{=} |\Lc_{\Bs}(r, \Cs, \Sc)| \quad \text{ for all } \Sc \text{ s.t. } \eqref{eq:S-cond} \text{ holds}\\
&\le |\bar{\Lc}_{\Bs}(r, \Sc)|,
\end{align*}
where $(a)$ follows from the fact that $\tilde \Cs := [\Cs^{\sf T}, \Is(\Sc)^{\sf T}]^{\sf T}$ is an invertible $L_\Bs\times L_\Bs$ matrix and $\Cs\Bs\Mub=\mathbf{0}$,
and thus, there is a one-to-one correspondence between $\Bs(\Sc) \Mub \Leftrightarrow\tilde \Cs \Bs \Mub \Leftrightarrow \Bs \Mub$.
The proof is a direct consequence of the fact that $\Bs(\Sc)\Mub$ is full rank and applying the following lemma on $\bar{\Lc}_{\Bs}(r, \Cs, \Sc)$.

\begin{lemma}
Consider a matrix $\Mub\in\Ic$, where the $k$-th row is $\mub_k\in[2^{n\Rt_k}]$.
Let $\Bs\in\Fq^{L_\Bs\times K}$, $1\le L_\Bs\le K$ be a full-rank matrix, and define
\begin{align*}
\Ac_{\Bs}(L_\Bs) = \{ \Mub_\Bs: \Mub_\Bs=\Bs\Mub,~ \Mub\in \Ic,~ \rank(\Mub_\Bs)=L_\Bs\}.
\end{align*}
Then,
\begin{align*}
|\Ac_{\Bs}(L_\Bs)| &\le \max_{\Tc}2^{n\Rt(\Tc)}
\end{align*}
where the maximum is over all $\Tc$ such that $|\Tc|=L_\Bs$ and $\rank([\Bs^{\sf T}, \Is(\Kc\setminus\Tc)^{\sf T}])=K$.
\end{lemma}

\begin{IEEEproof}
We will prove this upper bound by construction. First, we begin with a special case where
the rates are ordered by $\Rt_1 \ge \cdots \ge \Rt_K$ and $\Bs$ is in reduced row echelon form.

Let $\tilde\Tc = \{t_1, t_2, \ldots, t_{L_\Bs} \}$ be the set of pivot positions of $\Bs$.
Then the maximum number of non-zero entries in the $j$-th row of $\Mub_\Bs$ is 
the same as that in the $t_j$-th row of $\Mub$, which is given by
$\ceil{n \Rt_{t_j}  /\log(\q)}$ for $j = 1, \ldots, L_\Bs$.
In other words, the $j$-th row of $\Mub_\Bs$ has at most $2^{n\Rt_{t_j}}$ possibilities,
because the $t_j$-th row of $\Mub$ has at most $2^{n\Rt_{t_j}}$ possibilities.
Hence, $\Mub_\Bs$ has at most $2^{n\Rt(\tilde\Tc)}$ possibilities. This gives an upper 
bound for $|\Ac_{\Bs}(L_\Bs)|$ under the special case. Note that  $\tilde\Tc$ constructed above
satisfies the condition of $|\Tc|=L_\Bs$, $\rank([\Bs^{\sf T}, \Is(\Kc\setminus\Tc)^{\sf T}])=K$.
Therefore, 
\[
|\Ac_{\Bs}(L_\Bs)| \le 2^{n\Rt(\tilde\Tc)} \le \max_{\Tc}2^{n\Rt(\Tc)}.
\]
That is, the upper bound indeed holds for this special case.

Next, we consider a more general case where $\Rt_1 \ge \cdots \ge \Rt_K$ and $\Bs$ is not necessarily in reduced row echelon form.
Let $\mbox{RRE}(\Bs)$ be the reduced row echelon form of $\Bs$. Then 
$\mbox{RRE}(\Bs) = \Qb \Bs$ for some $L_\Bs \times L_\Bs$ invertible matrix $\Qb$.
Since $\Qb$ is invertible, the number of distinct $\Mub_\Bs$ is equal to the number of distinct $\mbox{RRE}(\Bs) \Mub$.
This reduces to our special case.

Finally, we consider the most general case where $\Rt_1,\ldots, \Rt_K$ can be in an arbitrary order. Then there exists a permutation $\pi : \mathcal{K} \to \mathcal{K}$
such that $\Rt_{\pi(1)} \ge \cdots \ge \Rt_{\pi(K)}$. In this case, we treat user $\pi(j)$ as our ``virtual" user $j$ and apply our previous argument to these virtual users.
In particular, we let $\tilde\Tc_{\pi}$ be the set of pivot positions of $\Bs$ with respect to the virtual users. Then, we have 
$|\Ac_{\Bs}(L_\Bs)| \le 2^{n\Rt(\tilde\Tc_{\pi})}$. Moreover, for any permutation $\pi$, $\tilde\Tc_{\pi}$ satisfies the condition of $|\Tc|=L_\Bs$, $\rank([\Bs^{\sf T}, \Is(\Kc\setminus\Tc)^{\sf T}])=K$.
Therefore, 
\[
|\Ac_{\Bs}(L_\Bs)| \le 2^{n\Rt(\tilde\Tc_{\pi})} \le \max_{\Tc}2^{n\Rt(\Tc)}.
\]
This completes the proof.
\end{IEEEproof}

\subsection{Proof of Lemma~\ref{lem:jt-lemma}}
Consider $1\le r \le L_\Bs$ and $\Cs \in\Fq^{(L_\Bs-r)\times L_\Bs}$ such that $\rank(\Cs)=L_\Bs-r$ and assume that $\Mub_\Bs \in \Lc_\Bs(r, \Cs)$.
Fix a set $\Sc\subseteq\Kc$ such that $|\Sc|=L_\Bs$ and
\begin{align}
\rank\left(\left[\Cs^{\sf T}, \Is(\Sc)^{\sf T}\right]\right)=L_\Bs.
\end{align}
Let $\tilde{\Ec}_1=\{U^n_k(0,0)\in\aep, k\in\Kc\}$.
Then, we have
\begin{align*}
\P&\{ (W_\Bs^n(\Mub_\Bs), Y^n)\in\aep, \Ec_1^c |\Mc \} \\
&\stackrel{(a)}{=}P\{(W_{\Bs(\Sc)}^n(\Mub_{\Bs(\Sc)}), W^n_{\Cs\Bs}, Y^n)\in\aep, \Ec_1^c |\Mc\} \\
&\stackrel{(b)}{\le} P\{(W_{\Bs(\Sc)}^n(\Mub_{\Bs(\Sc)}), W^n_{\Cs\Bs}, Y^n)\in\aep, \tilde{\Ec}_1|\Mc\} \\
&=\sum_{\uh_1^n\in\aep, \ldots, \uh_K^n\in\aep}\sum_{(w^n_{\Bs(\Sc)}, w^n_{\Cs\Bs}, y^n)\in\aep}  P\{W_{\Bs(\Sc)}^n(\Mub_{\Bs(\Sc)})=w^n_{\Bs(\Sc)}, U^n_\Kc(\zerob)=\uh_\Kc^n, W^n_{\Cs\Bs}=w^n_{\Cs\Bs}, Y^n=y^n |\Mc\} \\
&\stackrel{(c)}{=}\sum_{\uh_1^n\in\aep, \ldots, \uh_K^n\in\aep}\sum_{(w^n_{\Bs(\Sc)}, w^n_{\Cs\Bs}, y^n)\in\aep} P\{Y^n=y^n, W^n_{\Cs\Bs}=w^n_{\Cs\Bs} |U^n_\Kc(\zerob)=\uh_\Kc^n, \Mc\}\\
&\quad\times P\{W_{\Bs(\Sc)}^n(\Mub_{\Bs(\Sc)})=w^n_{\Bs(\Sc)}, U^n_\Kc(\zerob)=\uh_\Kc^n |\Mc\} \\
&\stackrel{(d)}{=}2^{n\Rh(\Kc)}\sum_{\uh_1^n\in\aep, \ldots, \uh_K^n\in\aep}\sum_{(y^n, w^n_{\Cs\Bs})\in\aep}p(y^n, w^n_{\Cs\Bs}|\uh_1^n, \ldots, \uh_K^n)\\
&\quad\times \sum_{w^n_{\Bs(\Sc)}\in\aep(W_{\Bs(\Sc)}| y^n,  w^n_{\Cs\Bs})}\P\{W_{\Bs(\Sc)}^n(\Mub_{\Bs(\Sc)})=w^n_{\Bs(\Sc)}, U^n_\Kc(\zerob)=\uh_\Kc^n \} \\
&\stackrel{(e)}{\le} 2^{n\Rh(\Kc)}\sum_{\uh_1^n\in\aep, \ldots, \uh_K^n\in\aep}\sum_{(y^n, w^n_{\Cs\Bs})\in\aep} p(y^n, w^n_{\Cs\Bs}|\uh_1^n, \ldots, \uh_K^n) 2^{n(H(W_{\Bs(\Sc)}|Y, W_{\Cs\Bs})+\d(\e))}\\
&\quad \times 2^{-n(H(W_{\Bs(\Sc)})+D(p_{W_{\Bs(\Sc)}}\| p^{|\Sc|}_\q)+\sum_{k\in\Kc}H(U_k)+\sum_{k\in\Kc}D(p_{U_k}\| p_\q)-\d(\e))} \\
&\le 2^{n\Rh(\Kc)}\sum_{\uh_1^n\in\aep, \ldots, \uh_K^n\in\aep}2^{n(H(W_{\Bs(\Sc)}|Y, W_{\Cs\Bs})+\d(\e))}\\
&\quad \times 2^{-n(H(W_{\Bs(\Sc)})+D(p_{W_{\Bs(\Sc)}}\| p^{|\Sc|}_\q)+\sum_{k\in\Kc}H(U_k)+\sum_{k\in\Kc}D(p_{U_k}\| p_\q)-\d(\e))} \\
&\le 2^{n\Rh(\Kc)}2^{n(H(W_{\Bs(\Sc)}|Y, W_{\Cs\Bs})+\d(\e))} 2^{-n(H(W_{\Bs(\Sc)})+D(p_{W_{\Bs(\Sc)}}\| p^{|\Sc|}_\q)+\sum_{k\in\Kc}D(p_{U_k}\| p_\q)-\d(\e))} \\
&= 2^{-n(I(W_{\Bs(\Sc)}; Y, W_{\Cs\Bs})+D(p_{W_{\Bs(\Sc)}}\| p^{|\Sc|}_\q)+\tilde{D}-\d(\e)))}
\end{align*}
where $U^n_\Sc(\mathbf{0})=(U^n_k(\mathbf{0}): k\in\Sc)$, step $(a)$ follows from the fact that $W^n_{\Bs}$ and $(W^n_{\Bs(\Sc)}, W^n_{\Cs\Bs})$ are deterministic functions of each other, step $(b)$ follows from the fact that $\aepvar\subseteq\aep$, step $(c)$ follows from the fact that conditioned on $\Mc$, $(Y^n, W^n_{\Cs\Bs})\to U^n_\Kc(\mathbf{0})\to W^n_{\Bs(\Sc)}(\Mub_{\Bs(\Sc)})$ forms a Markov chain, step $(d)$ follows from~\cite[Lemma 11]{Lim--Feng--Pastore--Nazer--Gastpar2018}, and step $(e)$ follows from~\cite[Lemma 7]{Lim--Feng--Pastore--Nazer--Gastpar2018}.

\bibliographystyle{IEEEtran}
\newcommand{\noopsort}[1]{}

\end{document}